\documentclass[prd,nofootinbib,preprintnumbers,preprint]{revtex4}
\usepackage{empheq}

\usepackage{amsmath}
\usepackage{amsfonts}
\usepackage{amssymb}
\usepackage{dsfont}
\usepackage[hyperfootnotes=false]{hyperref}
\usepackage[dvipsnames]{xcolor}

\newcommand{\X}{\mathcal X}
\newcommand{\K}{\mathcal K}

\newcommand{\G}{\mathcal G}
\newcommand{\mG}{\mathcal G}
\newcommand{\T}{\mathcal T}
\newcommand{\Q}{\mathcal Q}

\newcommand{\beq}{\begin{equation}}
\newcommand{\eeq}{\end{equation}}
\newcommand{\bAL}{\begin{align}}
\newcommand{\eAL}{\end{align}}

\newcommand{\LCM}{\nabla}

\newcommand{\LCS}{\mathcal D}
\newcommand{\SFS}{\chi}

\newcounter{example}[section]
\newcounter{remark}[section]
\newcounter{theorem}[section]
\newcounter{proposition}[section]
\newcounter{lemma}[section]
\newcounter{corollary}[section]
\newcounter{definition}[section]

\def\theremark{\arabic{section}.\arabic{remark}}
\def\thetheorem{\arabic{section}.\arabic{theorem}}

\def\thedefinition{\arabic{section}.\arabic{definition}}

\renewcommand*{\email}[1]{\footnote{Electronic address: \href{mailto:#1}{\nolinkurl{#1}} }}

\newenvironment{proof}{\noindent {\textit{Proof:}}
}{\medskip}

\newenvironment{theorem}{\refstepcounter{theorem}
\medskip\noindent{\bf Theorem \thetheorem}:\em}{\medskip}
\newenvironment{proposition}{\refstepcounter{theorem}\medskip\noindent{\bf
Proposition \thetheorem}:\em}{\medskip}

\newenvironment{definition}{\refstepcounter{definition}\medskip\noindent{\bf
Definition \thedefinition}:\em}{\medskip}

\begin{document}

\title{ 
Slice-reducible conformal Killing tensors, photon surfaces and shadows}
\author{Kirill Kobialko\email{kobyalkokv@yandex.ru}}
\author{Igor Bogush\email{igbogush@gmail.com}}
\author{Dmitri Gal'tsov\email{galtsov@phys.msu.ru}}
\affiliation{Faculty of Physics, Moscow State University, 119899, Moscow, Russia}

\begin{abstract}
We generalize our recent method for constructing Killing tensors of the second rank to conformal Killing tensors. The method is intended for foliated spacetimes of arbitrary dimension $m$, which have a set of conformal Killing vectors. It applies to foliations of a more general structure than in previous literature. The basic idea is to start with reducible Killing tensors in slices constructed from a set of conformal Killing vectors and the induced metric, and then lift them to the whole manifold. Integrability conditions are derived that ensure this, and a constructive lifting procedure is presented.  The resulting conformal Killing tensor may be irreducible. It is shown that subdomains of foliation slices suitable for the method are fundamental photon surfaces if some additional photon region inequality is satisfied. Thus our procedure also opens the way
to obtain a simple general analytical expression for the boundary of the gravitational shadow.
We apply this technique to electrovacuum, and ${\cal N}=2,\,4,\,8$ supergravity black holes, providing a new easy way to establish the existence of exact and conformal Killing tensors. 
\end{abstract}


\maketitle

\setcounter{page}{2}

\setcounter{equation}{0}
\setcounter{subsection}{0}
\setcounter{section}{0}

\section{Introduction}
Killing vectors describe spacetime isometries and generate first integrals of  geodesic  equations. Conformal Killing vectors describe spacetime symmetries preserving a line element up to a conformal factor and provide a set of integrals of motion for null geodesics. Additional symmetries, often referred to as hidden ones, are provided by Killing tensors (KT) and conformal Killing tensors (CKT) of higher rank. They correspond to phase space symmetries for geodesics and give higher-order conserved quantities in conjugate momenta. Basic facts about their existence and some methods of construction are known for more than half-century \cite{Carter:1968ks,Walker:1970un,Hughston:1972qf,sommers,benenti, Carter:1977pq, Stephani:2003tm}, but this still remains an active area of research \cite{Gibbons:2011hg,Cariglia:2014ysa,Frolov:2017kze, Papadopoulos:2018nvd,Carson:2020dez,Papadopoulos:2020kxu}. 

KT and CKT hidden symmetries may follow from underlying stronger Killing-Yano symmetries in which case a complete theory was developed during the past decade, for a review see \cite{Frolov:2017kze}. Solutions admitting conformal Killing-Yano tensor belong to rather strict class of charged Kerr-NUT-(A)dS type $D$ solutions and their higher-dimensional generalizations related to type $D$ in classification of \cite{Coley:2004jv,Coley:2007tp}. Solutions admitting KT and CKT beyond this class are not fully classified, but their existence in some algebraically non-special spacetimes is known \cite{Keeler:2012mq,Chow:2014cca}. An old approach of Benenti and Francaviglia \cite{benenti}, based on a certain ansatz for inverse metric in stationary axisymmetric spacetimes  for which an irreducible KT can be presented explicitly,  was recently shown to be applicable to many of currently popular ``deformed'' Kerr metrics, splitting them on geodesically integrable and non-integrable \cite{Papadopoulos:2018nvd,Papadopoulos:2020kxu, Carson:2020dez}. 
Some partial results have been obtained for  special spacetimes with a warped/twisted product structure \cite{Krtous:2015ona},  spaces that admits a hypersurface orthogonal Killing vector field \cite{Garfinkle:2010er,Garfinkle:2013cha}, or special conformal Killing fields \cite{Koutras,Barnes:2002pm}. 
 
Somewhat surprisingly, the equations that determine the null geodesics and the integrability conditions for the second-rank KT equations turn out to be related \cite{Pappas:2018opz,Glampedakis:2018blj}. Furthermore, it has recently been shown \cite{Kobialko:2021aqg,Koga:2020akc} that the existence of KT can be related to important characteristics of spacetime, such as photon spheres and their generalizations: fundamental photon surfaces  \cite{Claudel:2000yi,Yoshino1,Yoshino:2019dty,Yoshino:2019mqw,Teo, Galtsov:2019fzq, Galtsov:2019bty,Cao:2019vlu, Koga:2022dsu,Lee:2020pre,Kobialko:2021uwy,Kobialko:2020vqf}, which are compact submanifolds paved by null geodesics. Such surfaces are important for strong gravitational lensing and black hole shadows \cite{Grover:2017mhm,Cunha:2018acu,Shipley:2019kfq,Grenzebach,Tsukamoto:2020iez,Tsukamoto:2020bjm,Tsukamoto:2021fsz,Tsukamoto:2021lpm,GrenzebachSBH,Grenzebach:2015oea,Konoplya:2021slg,Perlick:2021aok,Lan:2018lyj}. They are equally useful in the analysis of the black hole uniqueness \cite{Cederbaum,Yazadjiev:2015hda,Yazadjiev:2015mta,Yazadjiev:2015jza,Yoshino:2016kgi,Yazadjiev:2021nfr,Koga:2020gqd,Rogatko,Cederbaumo,Cederbaum:2019rbv} and in deriving the area bounds \cite{Shiromizu:2017ego,Feng:2019zzn,Yang:2019zcn}. 

The novel geometric method for constructing KT was proposed in Ref. \cite{Kobialko:2021aqg}. It is applicable to manifolds with foliation of codimension one and is based on lifting reducible KTs constructed in slices from Killing vectors tangent to them. This approach does not require the slices to be orthogonal to the Killing vector field as was assumed in \cite{Garfinkle:2010er,Garfinkle:2013cha}, nor does it assume that the spacetime has a warped/twisted product structure as was considered in \cite{Krtous:2015ona}.  

The purpose of this article is to generalize this approach to CKT, as well as to demonstrate an important connection with the structure of fundamental photon surfaces and the formation of gravitational shadows. We find general integrability conditions for slice-to-bulk lift equations and prove the theorem \ref{KBG}, which provides a constructive method for generating slice-reducible CKTs that may be bulk irreducible. We discuss connection of such slice-reducible CKTs with photon surfaces and find a compact general formula for the shadow boundary of black holes applicable in many important solutions of various models. We apply this new technique to Plebanski-Demianski, EMDA and STU black holes solutions showing that this technique allows obtaining CKT \cite{Kubiznak:2007kh,Vasudevan:2005bz} purely algebraically without solving any differential equations and find the corresponding gravitational shadows. We introduce  the classes of slice-reducible KTs and CKTs as a promising tool for obtaining further analytical results. 

The work plan is the following.
In Sec. \ref{sec:setup} we briefly describe the equations for the conformal Killing vectors and second-rank CKT, and also briefly consider the formalism and notation of the general spacetime foliation of codimension one. In Sec. \ref{sec:conformal_killing_tensor} we present general lifting equations and  formulate the main theorems for conformal and exact KTs. In Sec. \ref{sec:photon}   the relationship between CKT and fundamental photon surfaces and shadows is discussed. In Sec. \ref{sec:examples} we apply formalism to many important particular cases. The appendices contain proofs of the statements formulated in the main part of the article. 

\section{Setup} 
\label{sec:setup}
Let $M$ be a Lorentzian manifold \cite{Chen} of dimension $m$ with metric tensor $g_{\alpha\beta}$ and Levi-Civita connection $\LCM_\alpha$ in some holonomic basis. 

\begin{definition} 
Tensor fields $K^\alpha$ and $K_{\alpha\beta}$ are called conformal Killing vectors and conformal Killing tensors respectively if \footnote{Here, the symmetrization is defined with weight one: $T_{(\alpha\beta)}=T_{\alpha\beta}+T_{\beta\alpha}$.} 
\begin{subequations}\label{eq:killing_equation}
    \begin{align} \label{eq:killing_equation_a}
        &
        \LCM_{(\alpha} K_{\beta)} = \Omega g_{(\alpha\beta)},
        \qquad
        \Omega = \frac{1}{m}\LCM_{\alpha} K^{\alpha},
        \\ & \label{eq:killing_equation_b}
        \LCM_{(\alpha} K_{\beta\gamma)} = \Omega_{(\alpha} g_{\beta\gamma)},
        \qquad
        \Omega_{\gamma} = 
        \frac{1}{m+2} \LCM_{\alpha}\left(
              2K^{\alpha}_{\gamma}
            + K\delta^\alpha_{\gamma}
        \right),
    \end{align}
\end{subequations}
where $K\equiv K^\alpha_\alpha$ (see Refs. \cite{Walker:1970un,Frolov:2017kze}).
\end{definition} 

Conformal Killing vectors and CKTs correspond to a set of conserved quantities defined as $\Q_{(1)}\equiv\dot{\gamma}^\alpha K_{\alpha}$ and $\Q_{(2)}\equiv\dot{\gamma}^\alpha \dot{\gamma}^\beta K_{\alpha\beta}$ for an arbitrary null geodesic $\gamma$ with affine parametrization, i.e., $\dot{\gamma}^\alpha\LCM_\alpha \Q_{(1,2)}=0$. One can always construct a CKT from the conformal Killing vectors \cite{Walker:1970un,Barnes:2002pm}. The corresponding conserved quantity is not independent, as far it can be expressed via conserved quantities associated with the conformal Killing vectors. Such reducible CKTs can be constructed as symmetrized products by following procedure:

\begin{proposition}\label{prop:trivial_tensor}
Let $K_a{}^{\alpha}$ be a set of $n$ conformal Killing vectors fields ($a=1,\ldots,n$). Then, one can define a reducible CKT field $K^{\alpha\beta}$ as follows 
\begin{subequations} \label{eq:trivial_tensor}
    \begin{align}
        & \label{eq:trivial_tensor_a}
        K^{\alpha\beta} =
          \alpha g^{\alpha\beta}
        + \gamma^{ab} K_a{}^{\alpha} K_b{}^{\alpha}, 
        \\& \label{eq:trivial_tensor_b}
        \Omega^{\alpha} =\LCM^\alpha \alpha+ 
        2 \gamma^{ab} \Omega_a K_b{}^{\alpha}, \quad \Omega_a=m^{-1}\LCM_{\alpha} K_a{}^{\alpha}, 
    \end{align}
\end{subequations}
where $\alpha$ is an arbitrary function and $\gamma^{ab}$ is a symmetric matrix filled with a set of $n(n+1)/2$ independent constants. 
\end{proposition}

\begin{proof}
The conformal Killing equation (\ref{eq:killing_equation_a}) applied to the tensor (\ref{eq:trivial_tensor_a}) gives the identity
\begin{align}
    \LCM_{(\alpha} K_{\beta\gamma)} =  \LCM_{(\alpha} (\alpha) g_{\beta\gamma)}+
    2 \gamma^{ab} \LCM_{(\alpha}(K_{a\beta})K_{b\gamma)} =\LCM_{(\alpha} (\alpha) g_{\beta\gamma)} +
    2 \gamma^{ab} \Omega_a g_{(\alpha\beta}K_{b\gamma)}.
\end{align}
Comparing this with Eq. (\ref{eq:killing_equation_b}) the corresponding vector $\Omega^{\alpha}$ reads as
\begin{align}
\Omega^{\alpha}=\LCM^\alpha \alpha+ 2 \gamma^{ab} \Omega_a K_b{}^{\alpha}.
\end{align}
\end{proof}
 
However, manifolds may have {\it irreducible} CKTs, providing new independent conserved quantities for null geodesics which refer to the hidden symmetries of geometries \cite{Frolov:2017kze}. The procedure of finding irreducible CKTs is not so straightforward as for reducible tensors and conformal Killing vectors. To simplify this problem, there exist various methods  based on the additional rigid restrictions  on the geometry of the manifold \cite{Krtous:2015ona,Garfinkle:2010er,Garfinkle:2013cha}. In this paper, we consider spacetimes with less restrictive conditions, namely manifolds with a codimension one foliation of general form that have additional arbitrary isometries.

In order to foliate a spacetime, we will use the following formalism and notations. First, we choose a foliation function $\Phi$ which can always be defined locally for any foliation. The vector normal to the foliation slices $S_s$ ($\Phi=s$) is
\begin{equation} \label{eq:hyper.splitting}
  n_\alpha=\pm\varphi\cdot\LCM_\alpha \Phi, \qquad n^\alpha n_\alpha=\epsilon,
\end{equation}
where $\varphi\equiv||\LCM_\alpha \Phi||^{-1}$ is the lapse function, and $\epsilon$ is $+1$ for timelike slices and $-1$ for spacelike slices. The induced slice metric reads
\begin{equation}
   h_{\alpha\beta} =
    g_{\alpha\beta}
    -\epsilon n_\alpha n_\beta,
\end{equation}
defining the projector operator $h^{\alpha}_{\beta}$ and the symmetric second fundamental form $\SFS_{\alpha\beta}$\footnote{The convention for second fundamental form $\chi$ is related to $\sigma$ in Ref. \cite{Kobialko:2021aqg} as follows: $\chi_{\alpha\beta} X^\alpha Y^\beta \equiv -\epsilon \sigma(X,Y)$.}
\begin{equation} \label{eq:form.projector}
    h^{\alpha}_{\beta} =
    \delta^{\alpha}_{\beta}
    -\epsilon n^\alpha n_\beta,\qquad
    \SFS_{\alpha\beta}
    \equiv
    h^{\lambda}_{\alpha}
    h^{\rho}_{\beta}
    \LCM_{\lambda} n_{\rho}.
\end{equation}
The corresponding tensor projections onto the tangent space of the slices read
\begin{equation} \label{eq:form.projector_2}
\T^{\beta\ldots}_{\gamma\ldots}=h^{\beta}_{\rho}\ldots  h_{\gamma}^{\tau}\ldots T^{\rho\ldots}_{\tau\ldots}, \qquad
    \LCS_\alpha \T^{\beta\ldots}_{\gamma\ldots} =
    h^{\lambda}_{\alpha} h^{\beta}_{\rho}\ldots  h_{\gamma}^{\tau}\ldots \LCM_\lambda \T^{\rho\ldots}_{\tau\ldots}, 
\end{equation}
where $\LCS_\alpha$ is a Levi-Civita connection in $S_s$. For a given foliation, we can additionally calculate the derivative of the normal vector along the slices flow
\begin{align} \label{eq:norm}
     n^\lambda \LCM_{\lambda} n_{\beta} =
    - \epsilon \LCS_\beta \ln \varphi.
\end{align}
In particular the commutator of the normal and tangent derivatives acting on a scalar is not zero (see App. \ref{auxiliary_identities})
\begin{align}
   &  [\varphi n^\alpha \nabla_\alpha, \LCS_\gamma] \omega =
    \left(
        n_\gamma \LCS^\alpha \varphi
        - \varphi \SFS^{\alpha}_{\gamma}
    \right)\LCS_\alpha \omega.   \label{eq:normal_tangent_commutator}
\end{align}
In what follows, the expression (\ref{eq:normal_tangent_commutator}) helps to reformulate the integrability condition for a system of first-order differential equations, which follows from the Frobenius theorem: 

\begin{proposition}\label{prop:integrability}
If there is an unknown scalar function $\omega$ such that 
\begin{equation} \label{eq:omega_eqs}
    n^\alpha \nabla_\alpha \omega = v,\qquad
    \LCS_\tau \omega = 0,
\end{equation}
then the integrability condition of this system of equations is
\begin{equation}
    \LCS_\alpha (\varphi v) = 0.
\end{equation}
\end{proposition}
\begin{proof}
A necessary and sufficient integrability condition on $\omega$ follows from the Frobenius theorem. From Eq. (\ref{eq:normal_tangent_commutator}) the vector field commutator acting on such function is zero. On the other hand, from Eq. (\ref{eq:omega_eqs}) follows another expression of the commutator
\begin{equation}
    [\varphi n^\alpha \nabla_\alpha, \LCS_\gamma] \omega =
    -\LCS_\gamma(\varphi n^\alpha \nabla_\alpha\omega) =
    -\LCS_\gamma(\varphi v),
\end{equation}
which should be zero. 
\end{proof}

\section{Conformal Killing tensors} 
\label{sec:conformal_killing_tensor}

The intended purpose of this section is to determine how the conformal Killing equations (\ref{eq:killing_equation}) split in general foliated spaces, generalizing the approach proposed in Refs. \cite{Garfinkle:2010er,Garfinkle:2013cha}.
In particular, we are interested in the appearance of an irreducible CKT lifted from a reducible one in a  foliation slice, as was the case with the KTs \cite{Kobialko:2021aqg}. This mechanism reflects the nature of hidden symmetries underlying the emergence of many well-known irreducible CKT fields. 

\subsection{Lift equations}
\label{sec:lift_equations}
\begin{proposition}\label{prop:projected_killing_equations}
Let $K^\alpha$ be a conformal Killing vector field in manifold $M$, foliated by $S_s$, with the normal $\zeta n^{\alpha}$ and tangent $\K^{\alpha}$ components, viz. $K^\alpha=\K^{\alpha} + \zeta n^{\alpha}$. Then, the Killing equations (\ref{eq:killing_equation_a}) split into the following parts (see App. \ref{proof_proposition_1}):
\begin{subequations}  \label{eq:kkv1}
\begin{align}
    & \label{eq:kkv1_1}
    \LCS_{(\alpha}\K_{\beta)} + 2 \zeta \SFS_{\alpha\beta} =
    2 \Omega h_{\alpha\beta},
    \\ &  \label{eq:kkv1_2}
    h^{\alpha}_{\rho} n^{\beta}\LCM_{\beta}\K_{\alpha} =
    \SFS^{\beta}_{\rho}\K_{\beta}+\epsilon \zeta  \LCS_{\rho} \ln (\varphi/\zeta),
    \\ & \label{eq:kkv1_3}
      \K^{\beta}\LCS_{\beta}\ln \varphi
    + n^{\alpha}\LCM_{\alpha}\zeta =
    \Omega.
\end{align}
\end{subequations}
\end{proposition}

As a consequence of Eq. (\ref{eq:kkv1_1}), the conformal Killing vector field in $M$, tangent to all foliation slices $S_s$ (i.e., $\zeta=0$), is also a conformal Killing vector field in $S_s$ with the same function $\Omega$. In a general case of non-tangent vector fields, the projection $\mathcal{K}_\alpha$ is a conformal Killing vector in the foliation slices only if the second fundamental form $\chi_{\alpha\beta}$ is proportional to the induced metric $h_{\alpha\beta}$, i.e., slices are totally umbilic. Such totally umbilic slices arise if they are generated by a set of (conformal) Killing vector fields and/or the spacetime has a structure of the warped/twisted product \cite{Chen}. Therefore, the generation of the non trivial conformal Killing vector from low-dimensional one is possible only in the case of the totally umbilic foliation. As we will see further, the case of the CKTs is more intricate. 

\begin{proposition} \label{prop:killing_map}
Let $K_{\alpha\beta}$ be a CKT field in manifold $M$, foliated by $S_s$, with the normal $\zeta n_{\alpha}n_{\beta}$, tangent $\K_{\alpha\beta}$ and mixed $\beta_\alpha$ ($n^\alpha \beta_\alpha=0$) components, viz. $K_{\alpha\beta}=\K_{\alpha\beta} + \zeta n_{\alpha} n_{\beta} + \beta_{(\alpha}n_{\beta)}$. Then, the Killing equations (\ref{eq:killing_equation_b}) split into the following parts (see App. \ref{proof_proposition_2}):
\begin{subequations} \label{eq:killing_map_xy}
\begin{align}
    &
    \LCS_{(\alpha} \K_{\beta\gamma)}
    + 2 \beta_{(\alpha} \SFS_{\beta}{}_{\gamma)} =
    {}^\tau\Omega_{(\alpha} h_{\beta\gamma)},
   \label{eq:killing_map_xy_a} \\ &
    h^\beta_{(\sigma} h^\gamma_{\tau)} n^\alpha \LCM_{\alpha} \K_{\beta\gamma}
    - 2\SFS^\alpha_{(\sigma} \K_{\tau)\alpha}
    + 2 \epsilon \zeta \chi_{(\sigma\tau)}
    + 2 \epsilon \LCS_{(\sigma}\beta_{\tau)}
    - 2 \epsilon \beta_{(\sigma} \LCS_{\tau)} \ln \varphi =
    {}^n\Omega h_{(\sigma\tau)},
    \label{eq:killing_map_xy_b} \\ &
      h^\gamma_\tau n^\alpha \LCM_{\alpha}\beta_{\gamma}
    + \K^\beta_{\tau}\LCS_\beta \ln \varphi 
    + \epsilon\LCS_\tau (\zeta) /2
    - \epsilon \zeta \LCS_{\tau} \ln \varphi
    - \SFS^\alpha_\tau\beta_{\alpha} =
      {}^{\tau}\Omega_{\tau}/2,
    \label{eq:killing_map_xy_c} \\ &
      n^\alpha \LCM_{\alpha} \zeta
    + 2 \beta^{\beta} \LCS_\beta \ln \varphi =
    \epsilon \cdot {}^n\Omega, \label{eq:killing_map_xy_d}
\end{align} 
\end{subequations}
where
\begin{equation}
    {}^\tau\Omega_{\alpha}=h^\beta_\alpha\Omega_{\beta}, \quad 
    {}^n\Omega=n^\beta\Omega_{\beta}.
\end{equation}
\end{proposition}


As before, Eq. (\ref{eq:killing_map_xy_a}) reduces to the conformal Killing equation in slices only under some additional assumptions. 

\begin{proposition} \label{prop:killing_mixed} 
Tangent projections $\K_{\alpha\beta}$ of the CKT $K_{\alpha\beta}$ is a CKT in slices if and only if at least one of the following conditions is fulfilled (see App. \ref{proof_proposition_3}):

(a) Zero mixed components, i.e., $\beta_{\alpha}=0$.

(b) Slices are totally umbilic, i.e., $\chi_{\alpha\beta}\sim h_{\alpha\beta}$.  

(c) Dimension of spacetime is $m=3$, $\beta_{\alpha}\beta^{\alpha}=0$ and
\begin{align} 
&h_{\alpha\beta}=\alpha_{(\alpha} \beta_{\beta)}/(\alpha_{\gamma}\beta^{\gamma}), \quad \sigma_{\alpha\beta}= \alpha_{(\alpha}   \alpha_{\beta)}/(\alpha_{\gamma}\beta^{\gamma}),
\end{align}
where $\alpha_\alpha$ is an arbitrary null vector such that $\alpha_{\gamma}\beta^{\gamma}\neq0$.
\end{proposition}


In what follows we ignore the case (c) since we are mainly interested in spacetimes of higher dimensions $m>3$. The case of totally geodesic/umbilic slices (b) was considered, to some extent, in Refs. \cite{Krtous:2015ona,Garfinkle:2010er,Garfinkle:2013cha}. Since in the framework of the current paper we are mainly interested in non umbilic slices we consider only the first case (a), and Eq. (\ref{eq:killing_map_xy}) reduces to the following system
\begin{subequations} \label{eq:kkv2}
\begin{align}
    &
    \LCS_{(\alpha} \K_{\beta\gamma)} =
    {}^\tau\Omega_{(\alpha} h_{\beta\gamma)},
    \\ &
    h^\beta_{(\sigma} h^\gamma_{\tau)} n^\alpha \LCM_{\alpha} \K_{\beta\gamma}
    - 2\SFS^\alpha_{(\sigma} \K_{\tau)\alpha}
    + 2 \epsilon \zeta \chi_{(\sigma\tau)} =
    {}^n\Omega h_{(\sigma\tau)},
    \\ &
    \LCS_\tau \zeta =
      \epsilon\cdot{}^{\tau}\Omega_{\tau}
    + 2\zeta \LCS_{\tau} \ln \varphi
    - 2\epsilon \cdot\K^\beta_{\tau}\LCS_\beta \ln \varphi,
    \\ &
    n^\alpha \LCM_{\alpha} \zeta =
    \epsilon \cdot {}^n\Omega.   
\end{align} 
\end{subequations}
In the general case, the problem of finding a solution of the equations (\ref{eq:kkv2}) still remains very nontrivial. However, the situation can be simplified if the slices of the foliation have a number of additional internal symmetries. 

\subsection{Slice-reducible CKT}
\label{sec:slice_reducible}

Suppose that the manifold $M$ has a collection of $n\leq m-2$ linearly independent conformal Killing vector fields $K_a{}^\alpha$ tangent to the foliation slices $S_s$, i.e., $K_a{}^\alpha=\K_a{}^\alpha$. The Gram matrix $\G_{ab}\equiv \K_a{}_{\alpha} \K_b{}^{\alpha}$ is assumed to be invertable $\G^{ab} = \left(\G_{ab}\right)^{-1}$. Additionally, we will introduce a set of complementary vectors $\mathcal{X}_i^\alpha$ ($i=1,\ldots,m-n-1$) orthogonal to Killing vectors $\mathcal{K}_a{}_\alpha \mathcal{X}_i^\alpha  = 0$ such that $\mathcal{K}_a{}^\alpha$ and $\mathcal{X}_i^\alpha$ form a complete basis in the foliation slices\footnote{Such complete basis is well defined since one can naturally choose a set of linearly independent conformal Killing vectors such that their Gram matrix is not degenerate almost everywhere \cite{Kobialko:2021aqg}.}. As a result of proposition \ref{prop:projected_killing_equations}, such vectors $K_a{}^\alpha$ are also conformal Killing vectors in slices $S_s$, and a reducible CKT field of the form (\ref{eq:trivial_tensor}) is always defined
\begin{align} \label{eq:trivial_ansatz}
 \K_{\alpha\beta}=\alpha h_{\alpha\beta} + \gamma^{ab} \K_{a\alpha} \K_{b\beta}, \qquad
 \LCS_{\tau} \gamma^{ab}=0, \quad \gamma^{ab}=\gamma^{ba}.
\end{align} 
Substituting this reducible CKT (\ref{eq:trivial_ansatz}) into equations (\ref{eq:kkv2}), and using equations (\ref{eq:kkv1}) with $\zeta=0$ for Killing vector, we obtain the final system of equations
\begin{subequations} \label{eq:system_3}
\begin{align}
    &  \label{eq:system_3_2}
        \LCS_\tau \ln(\zeta-\epsilon \alpha)=\LCS_{\tau} \ln \varphi^2,
    \\ &  \label{eq:system_3_3}
        2\epsilon ( \zeta-\epsilon \alpha) \chi_{\sigma\tau} =
          \epsilon n^\alpha \LCM_{\alpha}( \zeta -\epsilon \alpha) h_{\sigma\tau}
        - n^\alpha \LCM_{\alpha}\gamma^{ab}\K_{a\sigma}\K_{b\tau},
    \\ &  \label{eq:system_3_1}
        {}^n\Omega= \epsilon n^\alpha \LCM_{\alpha} \zeta,
    \\ & \label{eq:system_3_4}
    {}^{\tau}\Omega_{\tau} =
    \LCS_\tau \alpha+ 2 \gamma^{ab}  \K_a{}_{\tau}  \K_b{}^{\alpha}\LCS_\alpha \ln \varphi,
\end{align} 
\end{subequations}
Equations (\ref{eq:system_3_1}), (\ref{eq:system_3_4}) give expressions for ${}^n\Omega$ and ${}^\tau\Omega_\tau$ through other quantities. There is always a reducible solution for these equations $\zeta=\epsilon \alpha$ and $\gamma^{ab}=\text{const}$ which corresponds to the reducible CKT in the bulk. However, in some cases it can also have irreducible solutions which corresponds to the "slice-reducible" CKT.

\begin{definition} 
Conformal Killing tensor $K_{\alpha\beta}$ are called slice-reducible if there exists at least one foliation of the manifold such that tangent projections  $\K_{\alpha\beta}$ in all slices are reducible.
\end{definition} 

Next, we will obtain  general explicit solutions for $\zeta$ and $\gamma^{ab}$, and find the corresponding necessary and sufficient integrability conditions. 

\subsection{Solution for $\zeta$}
\label{sec:solution_zeta}

Projection of Eq. (\ref{eq:system_3_3}) onto conformal Killing vectors $\K_{a}{}^\alpha$ and orthogonal vectors $\X_i^\alpha$ gives zero
\begin{equation} \label{eq:block_condition}
\SFS_{\alpha\beta}\K_{a}{}^\alpha \X_i^\beta=0,
\end{equation}
leading to the condition for the second fundamental form being block diagonal in the basis of vectors $\{\mathcal{K}_a, \mathcal{X}_i\}$. Contracting Eq. (\ref{eq:system_3_3}) with the orthogonal vectors $\X_{i}^\alpha$ only, one can find $\SFS_{\alpha\beta}$ to be partially umbilic \cite{Kobialko:2020vqf,Kobialko:2021aqg} in these directions
\begin{align} \label{eq:umbilic_condition}
\SFS_{\alpha\beta}\X_i^\alpha\X_j^\beta=\chi h_{\alpha\beta}\X_i^\alpha\X_j^\beta,
\end{align}
with a weighted trace of the second fundamental form over orthogonal directions $\mathcal{X}_i^\alpha$ \footnote{Such trace has the meaning of the principal or mean curvature of the given orthogonal subbundle. The quantity $\chi$ is related to $h{}^\Omega$ in Ref. \cite{Kobialko:2021aqg} as follows: $\chi = -\epsilon\; h{}^\Omega$}
\begin{align}\label{eq:chi_equation}
    \chi \equiv
    \frac{1}{m-n-1} \sum_{i} \frac{
        \mathcal{X}_i^\alpha \chi_{\alpha\beta} \mathcal{X}_i^\beta
    }{
        \mathcal{X}_i^\alpha h_{\alpha\beta} \mathcal{X}_i^\beta
    } =
    \frac{1}{2}n^\alpha \LCM_{\alpha} \ln(\zeta-\epsilon \alpha).
\end{align}      
Substituting $\zeta-\epsilon \alpha \equiv e^\Psi$ into Eq. (\ref{eq:chi_equation}) and (\ref{eq:system_3_2}), one can get
\begin{subequations}  \label{eq:kk_eq_psi}
\begin{align}
    &
    \label{eq:normal_eq_psi}
    n^\alpha \LCM_{\alpha}\Psi = 2\chi,
    \\ &
    \label{eq:tangent_eq_psi}
    \LCS_\tau \Psi=\LCS_{\tau} \ln \varphi^2.
\end{align}
\end{subequations}
The system can be simplified, introducing another associated quantity $\tilde\Psi\equiv \Psi - \ln \varphi^2 $ constant in slices
\begin{subequations} \label{eq:eq_psi}
\begin{align}
    &
    \label{eq:normal_eq_psi_t}
    n^\alpha \LCM_{\alpha}\tilde{\Psi} =
    2 \left(\chi - n^\alpha \LCM_{\alpha}\ln \varphi\right),
    \\ &
    \label{eq:tangent_eq_psi_t}
    \LCS_\tau \tilde{\Psi}=0.
\end{align}
\end{subequations}
Following proposition \ref{prop:integrability}, the system $(\ref{eq:eq_psi})$ is integrable if the following condition holds
\begin{align}
    \LCS_\gamma \left( \varphi \chi - \varphi n^\alpha \LCM_{\alpha} \ln \varphi \right)
    = 0.
\end{align}
Keeping in mind that function $\alpha$ is arbitrary in bulk, one can express $\zeta$ from the definition of $\Psi$
\begin{equation}
    \zeta =
    e^{\tilde\Psi} \varphi^2 + \epsilon \alpha.
\end{equation}

\subsection{Solution for $\gamma^{ab}$}
\label{sec:solution_gamma}

In order to find $\gamma^{ab}$ for arbitrary slices we need the following auxiliary statement which generalizes proposition 4.1 in Ref. \cite{Kobialko:2021aqg}.

\begin{proposition} \label{prop:second_form}
Let all conformal Killing vector fields $\K_a{}^\alpha$ be tangent to all slices $S_s$ of the foliation. Then the following identity holds (see App. \ref{proof_proposition_5}):
\begin{align}
 n^\alpha \nabla_\alpha \G_{ab}=2\SFS_{\alpha\beta} \mathcal{K}_{a}{}^{\alpha}\mathcal{K}_{b}{}^{\beta}.
\end{align} 
\end{proposition}


Projecting Eq. (\ref{eq:system_3_3}) onto directions along conformal Killing vectors $\K_{c}{}^\sigma \K_{d}{}^\tau$, and using proposition \ref{prop:second_form}, one will get the equation for $\gamma^{ab}$
\begin{align} \label{eq:gamma_kkv}
    n^\alpha\LCM_{\alpha} \left(
          e^{\Psi} \G^{ab} - \epsilon \gamma^{ab}
    \right) = 0,
\end{align}
in addition to the constancy on slices $\LCS_\tau \gamma^{ab} = 0$. For the sake of convenience, we represent the solution in the form
\begin{align}
    \gamma^{ab} = \epsilon e^{\Psi} \G^{ab} + \nu^{ab},\qquad
    n^\alpha\LCM_{\alpha} \nu^{ab} = 0,
\end{align} 
where $\nu^{ab}$ is a new matrix. Following proposition \ref{prop:integrability}, the integrability condition for $\gamma^{ab}$ is
\begin{equation} \label{eq:integrability_B}
    \LCS_\gamma( \varphi n^\alpha \nabla_\alpha (e^\Psi \G^{ab})) = 0.
\end{equation}
Using the Leibniz rule and Eq. (\ref{eq:normal_eq_psi}), the condition (\ref{eq:integrability_B}) can be rewritten in the following $\Psi$-free form
\begin{equation}\label{eq:integrability_B_2}
    \LCS_\gamma\left(
          \G^{ab}
        + \frac{1}{2\chi}  n^\alpha \nabla_\alpha \G^{ab}
    \right) 
    =
    -\LCS_\gamma\ln\left(\chi \varphi^3\right)
    \left(
          \G^{ab}
        + \frac{1}{2\chi}  n^\alpha \nabla_\alpha \G^{ab}
    \right),
\end{equation}
reminding the integrability condition in Ref. \cite{Kobialko:2021aqg}. Though condition (\ref{eq:integrability_B}) is fair for any $\chi$, condition (\ref{eq:integrability_B_2}) does not work for $\chi=0$. It is so, since we implicitly assume $\chi\neq0$, deriving (\ref{eq:integrability_B_2}) from (\ref{eq:integrability_B}). Going through the same steps in the derivation of (\ref{eq:integrability_B_2}), but taking into account $\chi=0$, results in another condition
\begin{equation} \label{eq:integrability_B_3}
    \LCS_\gamma(n^\alpha \nabla_\alpha  \G^{ab})
    =
    - \LCS_\gamma\ln\varphi^3 \cdot n^\alpha \nabla_\alpha  \G^{ab}.
\end{equation}
Thus, we have obtained all unknown functions of the system (\ref{eq:system_3}) and can formulate the following theorem representing the key result.

\subsection{The theorem}
\label{sec:KBG_theorem}

\begin{theorem} 
\label{KBG}
Let the manifold $M$ contains a collection of $n\leq m-2$ conformal Killing vector fields $\K_a{}^\alpha$ with a non-degenerate Gram matrix $\G_{ab}=\K_a{}^\alpha\K_b{}_\alpha$, tangent to the foliation slices $S_s$ with the second fundamental form $\chi_{\alpha\beta}$ satisfying (\ref{eq:block_condition}), (\ref{eq:umbilic_condition}) and the integrability conditions\footnote{If $\chi=0$, Eq. (\ref{eq:integrability_G_b}) should be replaced by  (\ref{eq:integrability_B_3}).}
\begin{subequations} \label{eq:integrability_G}
\begin{align} \label{eq:integrability_G_a}
  &\LCS_\gamma \left( \varphi \chi - \varphi n^\alpha \LCM_{\alpha} \ln \varphi \right)
    = 0,\\ 
&  \label{eq:integrability_G_b} \LCS_\gamma\left(
          \G^{ab}
        + \frac{1}{2\chi}  n^\alpha \nabla_\alpha \G^{ab}
    \right) 
    =
    -\LCS_\gamma\ln\left(\chi \varphi^3\right)
    \left(
          \G^{ab}
        + \frac{1}{2\chi}  n^\alpha \nabla_\alpha \G^{ab}
    \right).
\end{align}
\end{subequations}
Then, there is a slice-reducible CKT on manifold $M$, which can be constructed as follows:

{\bf Step one}: Obtain $\Psi = \ln \varphi^2 + \tilde\Psi$ from the equations
\begin{align} \label{eq:equations_Psi}
     n^\alpha \LCM_{\alpha}\tilde{\Psi} =
    2 \left(\chi - n^\alpha \LCM_{\alpha}\ln \varphi\right), \quad \LCS_\tau \tilde \Psi = 0.
\end{align}

{\bf Step two}: Obtain $\gamma^{ab}$ from the equations
\begin{align}
    \gamma^{ab} = \epsilon e^{\Psi} \G^{ab} + \nu^{ab}, \quad
    n^\alpha\LCM_{\alpha} \nu^{ab} = 0, \quad \LCS_\tau \gamma^{ab} = 0.
\end{align} 

{\bf Step three}: Using the functions found in the previous steps, construct the corresponding CKT:
\begin{subequations}
\begin{align}
    \label{eq:theorem_tensor}
   & K_{\alpha\beta} =
    \alpha g_{\alpha\beta}
    + \gamma^{ab} \K_{a\alpha} \K_{b\beta}
    + e^{\Psi} n_{\alpha} n_{\beta},\\
   & \label{eq:theorem_factor}
   {}^n\Omega =n^\alpha \LCM_{\alpha} \alpha+ 
   2 \epsilon \chi  e^{\Psi}, \quad {}^{\tau}\Omega_{\tau}=\LCS_\tau \alpha+ 2 \gamma^{ab}  \K_a{}_{\tau}  \K_b{}^{\alpha}\LCS_\alpha \ln \varphi,
\end{align} 
\end{subequations}
where $\alpha$ is an arbitrary function.
\end{theorem}

Instead of this algorithm, one can suggest the following modification. First, we check the integrability condition (\ref{eq:integrability_G_a}) and find $\Psi$ from step one. This allows us to use the condition (\ref{eq:integrability_B}) instead of (\ref{eq:integrability_G_b}), which works for all $\chi$. And then, we go through steps two and three as before.
\subsection{Umbilic foliation}
\label{sec:umbilic_foliation}
If the foliation slices turn out to be totally umbilic or geodesic, conditions (\ref{eq:block_condition}) and (\ref{eq:umbilic_condition}) are obviously satisfied automatically with $\chi=\chi_\alpha{}^\alpha/(m-1)$.
From proposition \ref{prop:second_form} and totally umbilic condition follows the relation between $\G^{ab}$ and $\chi$
\begin{align}
 n^\alpha \nabla_\alpha \G^{ab}=-2\chi\G^{ab}.
\end{align}   
Then, condition (\ref{eq:integrability_G_b}) (or (\ref{eq:integrability_B_3}) if $\chi=0$) is fulfilled automatically insofar as
\begin{align}
    \G^{ab} + \frac{1}{2\chi}  n^\alpha \nabla_\alpha \G^{ab}=0
    \qquad \text{(or} \quad
    n^\alpha \nabla_\alpha \G^{ab}=0\text{)},
\end{align}  
and from Eq. (\ref{eq:gamma_kkv}) follows the constancy of $\gamma^{ab}$ along the normal direction
\begin{align} 
  n^\alpha \nabla_\alpha (e^\Psi \G^{ab})=n^\alpha \nabla_\alpha \Psi e^\Psi \G^{ab}-2\chi e^\Psi \G^{ab}=0  \quad \Rightarrow  \quad  n^\alpha\LCM_{\alpha} \gamma^{ab} = 0.
\end{align}
That is, there is only a trivial solution $\gamma^{ab}=\text{const}$. Nevertheless, function $e^{\Psi}$ is not trivial. As the result, the obtained CKT consists of two terms: the reducible part $\alpha g_{\alpha\beta} + \gamma^{ab} \K_{a\alpha} \K_{b\beta}$ and $e^\Psi n^\alpha n^\beta$. Without loss of generality, we can neglect the reducible part, so the CKT is $K^{\alpha\beta} = e^\Psi n^\alpha n^\beta$ with $\Omega_\alpha = 2 \chi e^{\Psi} n_\alpha$. However, even the remaining CKT turns out to be reducible as a consequence of the following proposition.

\begin{proposition} \label{prop:reducible_umbilic}
If $K^{\alpha\beta} = e^\Psi n^\alpha n^\beta$ is a CKT with $\Omega_\alpha = 2 \chi e^{\Psi} n_\alpha$ and equations (\ref{eq:kk_eq_psi}) are satisfied then $e^{\Psi/2} n^\alpha$ is a conformal Killing vector (see App. \ref{proof_proposition_4}) 
\begin{align}  \label{eq:conformal_reducible}
 \LCM_{(\alpha} (e^{\Psi/2} n_{\beta)})= \Omega g_{(\alpha\beta)}, \quad  \Omega = \chi e^{\Psi / 2}.
\end{align}
\end{proposition}

Keeping in mind possibilities for the mixed component $\beta^\alpha$, i.e., proposition \ref{prop:killing_mixed}, theorem \ref{KBG} gives the general form and conditions for the existence of a slice-reducible CKTs in spaces with non-totally umbilic foliation which can be irreducible in the bulk, i.e., for the first option $\beta_\alpha = 0$. In contrast to this, in spaces with totally umbilic foliation, such CKTs are always reducible. However, in the latter case slice-reducible CKTs with $\beta^\alpha\neq0$ can also exist and it may turn out to be irreducible in the bulk \cite{Popa:2007eq}. We leave the analysis of this option as one of the directions for further researches. 
\subsection{Weyl invariance}
\label{sec:weyl_invariance}
The existence of the CKTs is an invariant property under the Weyl transformations. To check the consistency of theorem \ref{KBG}, let us show the invariance.
Let us perform Weyl transformation $\bar{g}_{\alpha\beta}=e^{2\psi}g_{\alpha\beta}$, keeping conformal Killing vectors $\K_a{}^\alpha$, the length $n_\alpha n^\alpha  = \epsilon$ and the foliation function invariant. The bar above symbols denotes the transformed quantities. Then one can obtain the following transformation rules
\begin{align} 
    \bar{n}^\alpha=e^{-\psi}n^\alpha, \quad
    \bar{\varphi}=e^{\psi}\varphi, \quad
    \bar{\chi}= e^{-\psi} \left(\chi + n^\alpha \LCM_\alpha \psi\right),\qquad
    \bar{\G}^{ab} = e^{-2\psi}\G^{ab}.
\end{align} 
The system of equations (\ref{eq:eq_psi}) for $\tilde{\Psi}$ remains invariant, so we can leave $\bar{\tilde{\Psi}} = \tilde{\Psi}$ invariant as well. Using the relation between $\Psi$ and $\tilde{\Psi}$, one can get a new expression for $\bar{\Psi} = 2\psi + \Psi$, so the quantity $e^{\bar{\Psi}} \bar{\G}^{ab}$ is invariant either. Following these transformation rules, one can come up with the invariance of the integrability conditions under Weyl transformations, which is consistent with the idea of CKTs.  

\subsection{Reduction to exact KT}
\label{sec:KBG_theorem_non}

The resulting tensor is exact KT if ${}^n\Omega = 0$, ${}^\tau\Omega_\tau = 0$. In Ref. \cite{Kobialko:2021aqg}, a similar theorem was shown to exist for KTs of rank two if $\K_{a}{}^\alpha$ are non-conformal Killing vectors. In this case $\K_a{}^\alpha \LCS_\alpha \ln \varphi = 0$, and Eq. (\ref{eq:theorem_factor}) becomes a new system of equations on function $\alpha$
\begin{align}
   n^\alpha \LCM_{\alpha} \alpha = - 2 \epsilon \chi  e^{\Psi}, \quad
   \LCS_\tau \alpha = 0.
\end{align}
The condition of the slice-constancy of $\alpha$ makes tensor $\K_{\alpha\beta}$ be a KT in slices. Following proposition \ref{prop:integrability}, the integrability condition for $\alpha$ coincides with the ``compatibility condition'' in Ref. \cite{Kobialko:2021aqg}
\begin{equation} \label{eq:integrability_nonconformal}
    \LCS_\gamma (\chi\varphi^3) = 0.
\end{equation}
Using this condition, the second integrability condition (\ref{eq:integrability_G_b}) reduces to the result in Ref. \cite{Kobialko:2021aqg}
\begin{equation}
    \LCS_\gamma\left(
          \G^{ab}
        + \frac{1}{2\chi}  n^\alpha \nabla_\alpha \G^{ab}
    \right) 
    = 0.
\end{equation}

The new condition in Eq. (\ref{eq:integrability_nonconformal}) allows us to formulate the theorem from Ref. \cite{Kobialko:2021aqg} more precisely.

\begin{theorem} 
\label{KBG_non}
Let the manifold $M$ contains a collection of $n\leq m-2$ Killing vector fields $\K_a{}^\alpha$ with a non-degenerate Gram matrix $\G_{ab}=\K_a{}^\alpha\K_b{}_\alpha$, tangent to the foliation slices $S_s$ with the second fundamental form satisfying (\ref{eq:block_condition}), (\ref{eq:umbilic_condition}) and the integrability conditions
\begin{subequations}
\begin{align}
   & \LCS_\gamma \left( \varphi \chi - \varphi n^\alpha \LCM_{\alpha} \ln \varphi \right)
    = 0,\\
   & \LCS_\gamma (\chi\varphi^3) = 0, \quad \LCS_\gamma\left(
          \G^{ab}
        + \frac{1}{2\chi}  n^\alpha \nabla_\alpha \G^{ab}
    \right) =0.
\end{align}
\end{subequations}
Then, there is a slice-reducible KT on manifold $M$,  which can be constructed as follows:

{\bf Step one}: Obtain $\Psi = \ln \varphi^2 + \tilde\Psi$ from equations
\begin{align}
     n^\alpha \LCM_{\alpha}\tilde{\Psi} =
    2 \left(\chi - n^\alpha \LCM_{\alpha}\ln \varphi\right), \quad \LCS_\tau \tilde \Psi = 0.
\end{align}

{\bf Step two}: Obtain $\alpha$ and $\gamma^{ab}$ from equations
\begin{subequations}
    \begin{equation}
       n^\alpha \LCM_{\alpha} \alpha = - 2 \epsilon \chi  e^{\Psi}, \quad
       \LCS_\tau \alpha = 0.
    \end{equation}
    \begin{equation}
        \gamma^{ab} = \epsilon e^{\Psi} \G^{ab} + \nu^{ab}, \quad
        n^\alpha\LCM_{\alpha} \nu^{ab} = 0, \quad \LCS_\tau \gamma^{ab} = 0.
    \end{equation} 
\end{subequations}

{\bf Step three}: Using the functions found in the previous steps, construct the corresponding KT:
\begin{align}
   & K_{\alpha\beta} =
      \alpha g_{\alpha\beta}
    + \gamma^{ab} \K_{a\alpha} \K_{b\beta}
    + e^{\Psi} n_{\alpha} n_{\beta}.  
\end{align} 
\end{theorem}

Note that the existence of the exact KT imposes more restrictions on the geometry of the foliation and includes conditions for the existence of a CKT. In particular, a metric that admits the existence of the KT always admits the existence of an infinite family of CKTs obtained by adding the quantity $\alpha g_{\alpha\beta}$.
\subsection{The case of $n\geq m-1$ conformal Killing vectors}
\label{sec:many_vectors}
If a space contains $n=m-1$ conformal Killing vectors with an invertable Gram matrix $\G^{ab} = \left(\G_{ab}\right)^{-1}$, we can also consider another generation scheme. Proceeding as before, we can find the following equations for the tensor lift in this particular case
\begin{subequations}
\begin{align} 
 &  \Psi=\tilde{\Psi}+\ln \varphi^2, \quad  \LCS_\tau\tilde{\Psi}=0, \\
       &\gamma^{ab} = \epsilon e^{\Psi} \G^{ab} + \nu^{ab}, \quad
    n^\alpha\LCM_{\alpha} \nu^{ab} = 0, \quad \LCS_\tau \gamma^{ab} = 0.
    \\ & 
        {}^n\Omega=n^\alpha \LCM_{\alpha} \alpha+ \epsilon e^{\Psi} n^\alpha \LCM_{\alpha} \Psi, \quad 
    {}^{\tau}\Omega_{\tau} =
    \LCS_\tau \alpha+ 2 \gamma^{ab}  \K_A{}_{\tau}  \K_b{}^{\alpha}\LCS_\alpha \ln \varphi.
\end{align}
\end{subequations}
Since there is no any directions $\X_i^\alpha$, equations that are obtained from the contraction with $\X_i^\alpha$ (\ref{eq:block_condition}), (\ref{eq:umbilic_condition}),  (\ref{eq:normal_eq_psi_t}) and (\ref{eq:integrability_G_a}) should be completely ignored. The only integrability condition left is (\ref{eq:integrability_B}).

If a space contains $n > m-1$ conformal Killing vectors, we can still use theorem \ref{KBG} just taking into consideration only $m-2$ or less elements with invertable Gram matrix or the aforementioned scheme with $m-1$ vectors. Choosing different subsets from the full set of conformal Killing vectors, the resulting slice-reducible CKTs are not guaranteed to be the same.
\section{Connection with photon submanifolds and shadows}
\label{sec:photon}
Formation of shadows and relativistic images of stationary black holes and other ultracompact objects is closely related to {\em photon regions} \cite{Grenzebach,Grenzebach:2015oea,GrenzebachSBH}, which are defined as compact domains where photons can travel endlessly without escaping to infinity or disappearing at the event horizon.
Indeed, the boundary of the gravitational shadow corresponds to the set of light rays that inspiral asymptotically onto the part of the spherical surfaces in photon regions \cite{Wilkins:1972rs,Teo:2020sey,Dokuchaev:2019jqq}.

Spherical surfaces in the photon region are just as important for determining the shadow of a stationary black hole as the {\em photon surfaces} \cite{Claudel:2000yi,Gibbons:2016isj} in the static case\footnote{For recent review of strong gravitational lensing and shadows see \cite{Perlick:2021aok,Cunha:2018acu,Dokuchaev:2019jqq}. }. Recall that an important property of the photon surfaces is established by the theorem asserting that these are timelike  totally umbilic hypersurfaces $S$ in spacetime. This property can serve as a constructive definition for analyzing photon surfaces instead of solving geodesic equations and plays a decisive role in in the analysis of the black hole uniqueness \cite{Cederbaum,Yazadjiev:2015hda,Yazadjiev:2015mta,Yazadjiev:2015jza,Yoshino:2016kgi,Yazadjiev:2021nfr,Koga:2020gqd,Rogatko,Cederbaumo,Cederbaum:2019rbv} and area bounds \cite{Shiromizu:2017ego,Feng:2019zzn,Yang:2019zcn}. It is especially useful in the cases when the geodesic equations are non-separable, and their analytic solution can not be found \cite{Cornish:1996de,Cunha:2016bjh,Semerak:2012dw, Shipley:2016omi,Cunha:2018gql,Cunha:2017eoe}.
 
However, in  rotating spacetimes such as Kerr, the surfaces $r=const$ in the photon  region do not fully satisfy the umbilic condition and may have a boundary.  
Such surfaces usually  form a family, specified by the value of the azimuthal impact parameter $\rho=L/E$, where $L,E$ are the integrals of motion corresponding to the timelike and azimuthal Killing vector fields \cite{Galtsov:2019bty,Galtsov:2019fzq}. To describe these surfaces and the photon  region geometrically, the concept of {\em partially umbilic} submanifolds that weaken the  umbilic  condition was introduced. Namely, it is possible to impose the umbilic condition not on {\em all} vectors from the tangent space $TS$, but only on some subset of $TS$ specified by the azimuthal impact parameter. In addition, we must specify the boundary conditions for the submanifolds so that the photon does not escape through them. Together, this leads to the definition of {\em fundamental photon  submanifolds} (FPS) \cite{Kobialko:2020vqf}. The slices consisting of the  fundamental photon surfaces form generalized photon regions. 
 
Here we will attempt to generalize   the relationship between fundamental photon surfaces and slice-reducible KTs found in  \cite{Kobialko:2021aqg}, generalizing the ideas of Ref. \cite{Koga:2020akc} where the relationship between KTs and ordinary photon surfaces was stated. 

\subsection{Photon submanifolds}
\label{sec:photon_submanifold} 
Recall the main ideas underlying the concept of FPS \cite{Kobialko:2020vqf,Kobialko:2021uwy}. Consider the case of a manifold with two conformal Killing vectors spanning a timelike surface ($\epsilon=+1$, $\det (\mathcal G_{ab})<0$). Let us define a conformal Killing vector field  
\begin{equation}\label{FPS1}
 \rho^\alpha=\rho^a\K_a{}^\alpha, \quad \rho^a=(\rho,1), 
\end{equation}
where $\rho$ is a constant and $\rho^a$ is not timelike anywhere ($\rho^\alpha\rho_\alpha \geq 0$). Consider an arbitrary affinely parameterized null geodesic $\gamma$ with the conserved quantity $\rho_\alpha\dot{\gamma}^\alpha$ equal to zero. From the definition (\ref{FPS1}) we find that $\rho=-\mathcal{K}_2{}_\alpha\dot{\gamma}^\alpha/\mathcal{K}_1{}_\alpha\dot{\gamma}^\alpha$ and $\rho$ can be called the \textit{generalized impact parameter} (see Ref. \cite{Kobialko:2020vqf} for details).
However, one can choose an arbitrary parametrization of $\rho^\alpha$ up to the norm. In addition, we introduce the orthogonal vector $\tau^a$ such that $\tau^a \G_{ab} \rho^b = 0$ (for simplicity we will use the definition $\tau^a \equiv \sqrt{-\G} \G^{ab}\epsilon_{bc} \rho^c $, where $\epsilon_{12}=-\epsilon_{21}=1$ is an antisymmetric tensor density).

\begin{definition}
The fundamental photon surface is a partially umbilic surface with a second fundamental form satisfying (\ref{eq:block_condition}), (\ref{eq:umbilic_condition}) such that for some $\rho^a\neq0$, $\tau_a\LCS_\alpha\rho^a=0$ (i.e., the direction of $\rho^a$ is constant at each surface) the following master equation is satisfied
\begin{align}
    \rho^a \mathcal{M}_{ab} \rho^b = 0,
    \qquad
    \mathcal{M}_{ab} \equiv
      \frac{1}{2 \chi}\cdot n^\alpha \LCM_\alpha \mG_{ab} - \frac{1}{2 \chi} n^\alpha \LCM_\alpha\ln \mathcal G\cdot\mG_{ab}
    + \mG_{ab}.
\label{FPS3}
\end{align}
\end{definition} 

Every FPS is characterized by its own $\rho$. The key property of FPS with a given $\rho$ is that any null geodesic with this $\rho$ on such a surface is a null geodesic in the bulk \cite{Kobialko:2020vqf}. If the photon surface is totally umbilic, the operator $\mathcal{M}_{\alpha\beta}$ is identically zero.

Consider the timelike foliation generating a slice-reducible CKT in accordance with the theorem \ref{KBG}. The question is whether such foliation slices form a set of fundamental photon surfaces. If the slices are totally umbilic, they represent photon surfaces by definition. Now consider the non-umbilic slices. One can easily find the solution of Eq. (\ref{FPS3}) with respect to $\rho$ and check the condition $\rho_\alpha \rho^\alpha\geq0$
\begin{subequations}
\begin{align}\label{PR1a}
    \rho^a = \left(\frac{-\mathcal M_{12}\pm \sqrt{-\mathcal M}}{\mathcal M_{11}},1\right),
\end{align}
\begin{align}\label{PR1b}
\pm 2(\mathcal G_{12}\mathcal M_{11}-\mathcal G_{11}\mathcal M_{12})\sqrt{-\mathcal M}-2\mathcal G_{11} \cdot \mathcal M+\mathcal M_{11} \cdot\mathcal G \cdot {\rm Tr}(\mathcal M) \geq0,
\end{align}
\end{subequations}
where 
\begin{equation}
\mathcal M\equiv \det (\mathcal M_{ab}) = \mathcal{M}_{11}\mathcal{M}_{22} - \mathcal{M}_{12}^2,\qquad
{\rm Tr}(\mathcal M)\equiv\mathcal M_{ab} \mathcal G^{ab}=2-(2\chi)^{-1}n^\alpha \LCM_\alpha\ln\mG.
\end{equation}
The inequality (\ref{PR1b}) defines the so-called photon region \cite{Grenzebach,Grenzebach:2015oea}, which arises as a flow of fundamental photon surfaces \cite{Kobialko:2020vqf}. Nevertheless, the constancy of the direction $\rho^a$ in each slice has not been proven yet.

Let us act on Eq. (\ref{FPS3}) with $\LCS_\tau$
\begin{align}
\LCS_\tau \mathcal{M}_{ab}\rho^a\rho^b+2\mathcal{M}_{ab}\rho^a \LCS_\tau\rho^b=0. 
\label{FPS4}
\end{align}
The expression from the second term $\LCS_\tau\rho^b$ consists of two parts: the change of the vector length $\left(\LCS_\tau\rho^b\right)_l\equiv(\rho_a\LCS_\tau\rho^a) / (\rho^a \rho_a) \cdot \rho^b$ directed along $\rho^b$, and the change of the vector direction $\left(\LCS_\tau\rho^b\right)_d\equiv(\tau_a\LCS_\tau\rho^a) / (\tau^a \tau_a) \cdot \tau^b$ directed along $\tau^b$. Since the master equation (\ref{FPS3}) holds, the change of the vector length $\left(\LCS_\tau\rho^b\right)_l$ does not contribute to Eq. (\ref{FPS4}). The term $\mathcal{M}_{ab}\rho^a \LCS_\tau\rho^b$ can be zero for non-zero $\left(\LCS_\tau\rho^b\right)_d$ only if $\mathcal{M}_{ab}\rho^a \tau^b = \mathcal{M}_{ab}\rho^a \rho^b = 0$, which is possible if and only if the matrix $\mathcal{M}_{ab}$ is degenerate.
The case of a degenerate matrix $\mathcal{M}_{ab}$ imposes a very strict constraint on $\G_{ab}$, $\chi$ and $n^\alpha \nabla_\alpha \G_{ab }$, which is not is of interest in this section. Thus, the slice-constancy condition for the direction $\rho^a$ is equivalent to the requirement $\LCS_\tau \mathcal{M}_{ab}\rho^a\rho^b=0$. 
Further, we will prove that this requirement is satisfied if the integrability condition (\ref{eq:integrability_B_2}) is satisfied. 

First, we rewrite the integrability condition (\ref{eq:integrability_B_2}) lowering the indices 
\begin{align}
    &
    \LCS_\gamma\left(
        \frac{1}{2\chi}  n^\alpha \nabla_\alpha \G_{ab}
    \right) 
    =
    \LCS_\gamma\ln\left(\chi \varphi^3\right) \cdot \left(
          \G_{ab}
        - \frac{1}{2\chi} n^\alpha \nabla_\alpha \G_{ab}
    \right)
    - \\\nonumber &
    - \LCS_\gamma\G_{ab}
    + \frac{1}{2\chi} \left(
          \LCS_\gamma\G_{ca}\cdot n^\alpha \nabla_\alpha \G_{bd}
        +\LCS_\gamma\G_{cb}\cdot n^\alpha \nabla_\alpha \G_{ad}
    \right)\G^{cd},
\end{align}
and plug it into the first term of Eq. (\ref{FPS4}) to get the following expression
\begin{align} \label{eq:dM}
&
\LCS_\tau \mathcal{M}_{ab}=\LCS_\tau  \left(\frac{1}{2 \chi}\cdot n^\alpha \LCM_\alpha \mG_{ab} - \frac{1}{2 \chi} n^\alpha \LCM_\alpha\ln \mathcal G\cdot\mG_{ab}
    + \mG_{ab}\right) \nonumber\\&
    =\LCS_\tau  \left(\frac{1}{2 \chi} n^\alpha \LCM_\alpha \mG_{ab}\right) - \LCS_\tau \left(\frac{1}{2 \chi}  n^\alpha  \LCM_\alpha \mathcal G_{cd}\right) \mathcal G^{cd} \mG_{ab}- \frac{1}{2 \chi}   n^\alpha  \LCM_\alpha \mathcal G_{cd} \cdot\LCS_\tau \left(\mathcal G^{cd}\mG_{ab}\right)
    + \LCS_\tau \mG_{ab}\nonumber\\&
    = \frac{1}{2\chi} \mathcal{N}_{ab}
    - \LCS_\tau\ln\left(\chi \varphi^3/\G\right) \mathcal{M}_{ab},
\end{align}
where we used the identities $n^\alpha  \LCM_\alpha \ln \mathcal G=\mathcal G^{ab}\cdot n^\alpha  \LCM_\alpha\mathcal G_{ab}$,  $\LCS_\tau\mathcal G_{ac}\cdot\mathcal G^{cb}=-\mathcal G_{ac}\cdot \LCS_\tau \mathcal G^{cb}$ and defined a new matrix $\mathcal{N}_{ab}$
\begin{align}
    &
    \mathcal{N}_{ab} \equiv
       n^\alpha  \LCM_\alpha \mathcal G_{cd} \cdot\LCS_\tau \mathcal G^{cd} \cdot \mG_{ab}
    - n^\alpha \nabla_\alpha \G_{bp} \cdot \LCS_\tau \G^{lp} \cdot \G_{la}
    - n^\alpha \nabla_\alpha \G_{ap} \cdot \LCS_\tau \G^{lp} \cdot \G_{lb}
        + \\\nonumber &
    + \LCS_\tau\ln\G \cdot n^\alpha \LCM_\alpha\ln \mathcal G\cdot\mG_{ab}
    -  n^\alpha  \LCM_\alpha \ln \mathcal G \cdot \LCS_\tau \mG_{ab}
    - \LCS_\tau\ln\G \cdot n^\alpha \LCM_\alpha \mG_{ab}.
\end{align}
The matrix $\mathcal{N}_{ab}$ is the same as in Ref. \cite{Kobialko:2021aqg}, which was proven to be identically zero. The second term of Eq. (\ref{eq:dM}) is equal to zero when contracted with $\rho^a\rho^b$. 
Thus, the vector $\rho^a$ maintains direction in each slice when the integrability condition (\ref{eq:integrability_B_2}) is satisfied. This allows us to formulate the following theorem. 
 
\begin{theorem}  \label{KBGFPS}
Let timelike foliation of the manifold $M$ satisfy all conditions of the theorem \ref{KBG} for $\text{dim} \{\mathcal{K}_\alpha\}=2$ and all slices $S_s$ have a compact spatial section. Then maximal subdomain $U_{PS} \subseteq S_s$ such that the inequality $\rho_\alpha \rho^\alpha \geq 0$ ($\rho_\alpha$ is determined for a given slice by the master equation (\ref{FPS3})) holds for all points in $U_{PS}$ is a fundamental photon surface\footnote{In the case of not compact spatial section, the slice is not a fundamental photon surface as defined in Ref. \cite{Kobialko:2020vqf}. However, the theorem can be generalized for such not compact surfaces too.}.
\end{theorem}

In particular, the region $U_{PR}\subseteq M$, such that the inequality (\ref{PR1b}) holds for all points in $U_{PR}$, is a photon region.

\subsection{Black holes shadows}
\label{sec:shadow}
As is known, $m=4$ spacetime with $n=2$ conformal Killing vectors and one irreducible CKT corresponds to a completely integrable dynamical system (for null geodesics). If the structure of the photon  region is known, then it is possible to obtain a general analytical expression for the shadow boundary \cite{Grenzebach,Grenzebach:2015oea,GrenzebachSBH,Perlick:2021aok,Konoplya:2021slg}. In this section, we consider the case of a slice-reducible CKT obtained using the theorem \ref{KBG}. 
To do that, consider an observer with four-velocity $v^\alpha=v^a \K_{a}{}^{\alpha}$ (stationary observer \cite{Pugliese:2018hju}) and all possible null geodesics  through it. The vectors tangent to these null geodesics are read
\begin{align}
    \dot{\gamma}^\alpha = 
         \mathfrak{N}( -v^{-1}v^\alpha
        + \sin(\Phi)\sin(\Theta) v^{-1} u^\alpha 
        + \cos(\Phi)\sin(\Theta) \X^{-1}\X^\alpha- \cos(\Theta) n^\alpha), \\\nonumber
        v=\sqrt{-v_\alpha v^\alpha}, \quad  \X=\sqrt{\X_\alpha \X^\alpha} \quad u^a=(-\G)^{1/2}\G^{ab}\epsilon_{bc} v^c,
\end{align}
where the coordinates $\Phi$, $\Theta$ encode the celestial sphere of the observer, the vector $\X$ is the only direction orthogonal to the conformal Killing vectors and tangent to the slices, and $\mathfrak{N}$ is some function.
The function $\mathfrak{N}$ and the corresponding impact parameter $\rho^\alpha$ are found from the conditions $\rho_\alpha \dot{\gamma}^\alpha=0$ (from the definition of $\rho$) and $\mathfrak{N}_\alpha \dot {\gamma}^\alpha=\Q_{\mathfrak{N}}$, where $\mathfrak{N}^\alpha$ is an arbitrary nonzero conformal Killing vector directed not along $\rho^\alpha$ and $\Q_{\mathfrak{N}}$ is the corresponding integral of motion: 
\begin{align}\label{eq:sinsin}
    &
        \sin(\Phi) \sin(\Theta) = \frac{\rho_{a} v^a}{\rho_{a} u^a},
    \quad
        \mathfrak{N}=v \Q_{\mathfrak{N}} \left(
              \frac{\rho_{a} v^a}{\rho_{a} u^a} u^a\mathfrak{N}_a
            - v^a \mathfrak{N}_a  
        \right)^{-1}.
\end{align}
Carter's constant \cite{Carter:1968ks} corresponding to the slice-reducible CKT has the form
\begin{align} \label{eq:carter_Q}
    &
        \Q_{(2)} \equiv
        K_{\alpha\beta} \dot{\gamma}^\alpha \dot{\gamma}^\beta =
        \mathfrak{N}^2 \left(
              e^{\Psi} \cos^2(\Theta)
            + v^{-2}\gamma^{ab} \Q_{ab}
        \right),
    \\\nonumber &
        \Q_{ab}\equiv v_a v_b
        -  \frac{\rho_{c } v^c}{\rho_{c} u^c} (v_a u_b+u_a v_b)
        +  \frac{(\rho_{c } v^c)^2}{(\rho_{c} u^c)^2} u_a u_b.
\end{align}
On the other hand, the shadow boundary corresponds to null geodesics asymptotically tangent to FPS \cite{Grenzebach,Grenzebach:2015oea,GrenzebachSBH} or ordinary photon surfaces \cite{Perlick:2021aok}. For such geodesics, $\dot{\gamma}^\alpha$ can be decomposed into a tetrad $\{n^\alpha, \X^\alpha, \rho^{\alpha} , \tau^\alpha\}$ (making the substitutions $v^\alpha\to\tau^\alpha$ and $u^\alpha\to\rho^\alpha$ in (\ref{eq:carter_Q})) which is always well-defined for FPS 
\begin{subequations}\label{eq:Q_FPS}
\begin{align}
\label{eq:Q_FPS_a}
 &\Q_{(2)}^{PS}= \Q^2_{\mathfrak{N}}(\tau_a \mathfrak{N}^a)^{-2}\gamma^{ab}\tau{}_{a}\tau_{a} \Big|_{PS}= \Q^2_{\mathfrak{N}}(\epsilon_{ab} \mathfrak{N}^a \rho^b)^{-2}\gamma^{ab}\epsilon_{ad} \epsilon_{bc} \rho^d_{PS}\rho^c_{PS}\Big|_{PS}, 
 \\\label{eq:Q_FPS_b} &\rho^{a}_{PS}=\left(\frac{-\mathcal M_{12}\pm \sqrt{-\mathcal M}}{\mathcal M_{11}},1\right)\Big|_{PS},
\end{align}
\end{subequations}
where $\Q_{(2)}^{PS}$ is Carter's constant computed for a given FPS. Comparing (\ref{eq:sinsin}), (\ref{eq:carter_Q}), (\ref{eq:Q_FPS}) and taking into account ${\Q_{(2)} = \Q_{(2)}^{PS}}$, the boundary of the shadow reads
\begin{subequations}\label{eq:general_boundary}
\begin{align}
 \label{eq:general_boundary_a}
 &
 \cos^2(\Theta_{SH}) =
 \left.\frac{e^{-\Psi} v^a v_a}{(\epsilon_{ab}v^a\rho_{PS}^b)^2}\right|_{O}
 \left(
      \left.\mathcal{J}_{ab} \right|_{O}
    - \left.\mathcal{J}_{ab} \right|_{PS}
 \right)
 \rho^a_{PS} \rho^b_{PS},
 \\& \label{eq:general_boundary_b}
 \sin(\Phi_{SH})\sin(\Theta_{SH}) =
 \left.\frac{v_a \rho_{PS}^{a}}{u_a\rho_{PS}^{a}}\right|_{O},
 \\& \label{eq:general_boundary_c}
   \mathcal J_{ab} \equiv
 - \epsilon_{ac}\gamma^{cd}\epsilon_{db} =
 \mathcal{J}^{(1)}_{ab} +\mathcal{J}^{(2)}_{ab},\qquad
 \mathcal{J}^{(1)}_{ab} \equiv e^\Psi \G_{ab}/\G,\qquad
 \mathcal{J}^{(2)}_{ab} \equiv - \epsilon_{ac}\nu^{cd}\epsilon_{db},
\end{align} 
\end{subequations}
where $\big|_{O}$ and $\big|_{PS}$ indicate that the expressions should be taken at the position of the observer and the FPS respectively. 

The expressions obtained specify the boundary as a function of the FPS foliation parameter and depend only on the geometry and four-velocity of the observer. 
These expressions are a generalization of similar analytical expressions for the shadow boundary from  Refs. \cite{Grenzebach,Grenzebach:2015oea,GrenzebachSBH,Perlick:2021aok}. 
According to the theorem \ref{KBG}, $\gamma^{ab}$ is constant for every FPS.
Therefore, $\mathcal{J}_{ab}$ is independent of the FPS slice point and can be computed anywhere in the corresponding slice. Although the FPS can only represent a subdomain of the entire slice of the foliation due to the condition $\rho_\alpha \rho^\alpha \geq 0$, the constancy of $\mathcal{J}_{ab}$ is true for the whole slice. For clarity, we introduce the coordinate $r$ corresponding to the foliation parameter and the coordinate $\theta$ along the direction $\X$. 
Also, according to the theorem \ref{KBG}, $\nu^{ab}$ does not depend on $r$
\begin{equation}
    \mathcal{J}_{ab}(r) = \mathcal{J}^{(1)}_{ab}(r,\theta) + \mathcal{J}^{(2)}_{ab}(\theta).
\end{equation}
Finally, to find $\mathcal{J}_{ab}$ we can use any value of $\theta$ in the expressions $\mathcal{J}^{(1,2)}_{ab}$. If so, $\left.\mathcal{J}_{ab}\right|_O$ and $\left.\mathcal{J}_{ab}\right|_{PS}$ can be computed at the same $\theta$
\begin{equation}\label{eq:JJ}
      \left.\mathcal{J}_{ab}\right|_O
    - \left.\mathcal{J}_{ab}\right|_{PS} =
      \mathcal{J}^{(1)}_{ab}(r_O,\theta)
    - \mathcal{J}^{(1)}_{ab}(r_{PS},\theta),
\end{equation}
where the final expression is a function of $r_O$ and $r_{PS}$, and one can substitute any value of $\theta$ in practical calculations (e.g., $\theta=\pi/2$). Thus, to determine the shadow, we need to know only the conformal Killing vector fields, the normal component of the slice-reducible CKT and the corresponding foliation with slices representing FPS.

\subsection{Shadows for asymptotically distant observers}
\label{sec:asymptotic_shadow}
Let us assume that manifold $M$ is a stationary locally asymptotically flat spacetime (with NUT $ N$)
\begin{align} \label{SolNUT_asymptotic}
   & \G_{ab} dx^a dx^b =
     - A(r,\theta)\left(dt - C(r,\theta) d\varphi\right)^2 + B(r,\theta) r^2\sin^2\theta d\varphi^2,
     \\\nonumber &
     \left.A\right|_{r\to\infty} \to 1,\qquad
     \left.B\right|_{r\to\infty} \to 1,\qquad
     \left.C\right|_{r\to\infty} \to C_\infty(\theta) = -2N(\cos\theta + C_{N}),\qquad
\end{align}
and the foliation is defined by the slices $r=\text{const}$ with $g_{rr}\to1$. An arbitrary stationary observer \cite{Pugliese:2018hju} 
\begin{align}\label{eq:alt_observer}
v^\alpha=v^t\delta^\alpha_t+v^\phi\delta^\alpha_\phi,\qquad
\left.v^t\right.|_{r\rightarrow \infty}\to 1, \quad  \left.r^2 v^\phi\right|_{r\rightarrow \infty}\to w(\theta_O),
\end{align}
floats at an asymptotically distant point $\theta=\theta_O$, $r=r_O\to\infty$ and have finite asymptotic angular momentum 
\begin{align}
\Q_{\varphi}|_{r\to\infty}\to \Q^{\infty}_{\varphi}\equiv C_\infty(\theta_O)+w(\theta_O)\sin^2\theta_O,
\end{align}
even for the case with no NUT. In particular, the asymptotically distant observer has zero angular momentum for $w(\theta_O)=-C_\infty(\theta_O)\sin^{-2}\theta_O$ \cite{Bardeen:1973tla,Pugliese:2018hju}. Applying all together and expanding the expressions in powers of $r_O^{-1}$, Eqs. (\ref{eq:general_boundary}) will look like 
\begin{subequations} \label{eq:flat_boundary}
\begin{align}
 &\cos^2(\Theta_{SH}) =1-\frac{R^2_{SH}}{r^2_O}+\mathcal{O}(r_O^{-3}),\quad \sin(\Phi_{SH})\sin(\Theta_{SH}) =\frac{\rho_{\mathcal N}}{r_O\sin\theta_O}+\mathcal{O}(r_O^{-2}),\label{eq:flat_boundary_a} \\&
\rho_{\mathcal N}\equiv \rho_{PS} -\Q^{\infty}_{\varphi}, \quad
R^2_{SH}\equiv\sin^{-2}\theta_O\cdot\rho^2_{\mathcal N}-\mathcal{J}^{(1)}_{ab}(r_{PS},\theta_O)\rho^a_{PS} \rho^b_{PS}, \label{eq:flat_boundary_b}
\end{align} 
\end{subequations}
where the expression (\ref{eq:JJ}) is evaluated with $\theta = \theta_O$ and $\rho_{PS}\equiv\rho^1_{PS}$. The limit of expressions taken at the position of the observer ($r_O\to\infty$) does not depend on the specific form of the functions $A,B,C$ and $v^t,v^\varphi$, except for their asymptotic values. The functions $A,B,C$ play an important role only when calculating $\rho_{PS}$ and $\mathcal{J}_{ab}^{(1)}(r_{PS}, \theta)$  or the higher order corrections. Since $\cos^2(\Theta_{SH})$ is close to 1, one can use the cosine expansion to get the following asymptotic solution 
\begin{align}
 &\Theta_{SH} =\frac{R_{SH}}{r_O}, \quad \Phi_{SH}=\arcsin{\left(\frac{\rho_{\mathcal N}(\theta_O)}{R_{SH}\sin\theta_O}\right)}.
\end{align}
The shadow is located near the zenith of the observer, where the coordinates $\Theta,\,\Phi$ can be locally considered as a polar coordinate system. The image of the shadow, parametrized by the Cartesian coordinate system, has the form 
\begin{subequations}
\label{eq:xy_shadow}
\begin{align}
    &
    X = -\frac{R_{SH}}{r_O}\sin(\Phi_{SH})
    = -\frac{\rho_{\mathcal N}}{r_O \sin\theta_O},
    \label{eq:xy_shadow_a}\\&
    Y=\pm\frac{R_{SH}}{r_O}\cos(\Phi_{SH})
    =\pm\frac{
        \sqrt{R_{SH}^2 \sin^2\theta_O - \rho_{\mathcal N}^2}
    }{r_O \sin\theta_O} 
    =\pm\frac{1
        }{r_O}\sqrt{-\mathcal{J}^{(1)}_{ab}(r_{PS},\theta_O)\rho^a_{PS} \rho^b_{PS}}
    .     \label{eq:xy_shadow_b}
\end{align}
\end{subequations}

To construct a shadow boundary with respect to some observer with a specific $\theta_O$, the following steps can be performed. Take some coordinate $r_{PS}$ and calculate the impact parameter $\rho_{PS}$ and $\mathcal{J}^{(1)}_{ab}(r_{PS},\theta_O)$ using ( \ref{eq:Q_FPS_b}) and (\ref{eq:general_boundary_c}). Substitute these values and the observer's azimuth $\theta_O$ into $\bar{X}=r_O X$, $\bar{Y}=r_O Y$ from (\ref{eq:xy_shadow}). Repeat these steps for all $r_{PS}$ determined from the condition $\left.\rho_{PS}^a \G_{ab} \rho_{PS}^b \right|^{r=r_{PS}}_{\theta=\theta_O}\geq0$ (subset in the full photon region surfaces family). 

From Eq. (\ref{eq:xy_shadow}) it is clear that changing the asymptotic value $w\to w'=w+\Delta w$ (and, accordingly, the asymptotically distant observer's angular momentum $\Q^{\infty}_{\varphi}$) does not change the shape of the shadow, but only shifts it  as
\begin{align}
\bar{Y} \rightarrow  \bar{Y},\quad \bar{X}\rightarrow \bar{X}+\Delta w(\theta_O) \sin \theta_O.
\end{align}
This shift can be interpreted as a different choice of the origin of the observer's celestial sphere \cite{Perlick:2021aok}. If the origin of the coordinate system is given by the principal null rays \cite{Grenzebach,Grenzebach:2015oea,GrenzebachSBH,Perlick:2021aok}, e.g., for Plebanski-Demianski solution $w(\theta)=a$ reproduces the result (53) from Ref. \cite{Perlick:2021aok} exactly. Thus, we have obtained a direct generalization of a number of expressions obtained in Refs. \cite{Perlick:2021aok,Konoplya:2021slg}.

\section{Applications}
\label{sec:examples}
Our framework opens up a way to combine most of the previously known results  in a unique algorithm that shows their relationship to the photon structure of manifolds. This includes the most general type $D$ vacuum and electrovacuum solutions as well as type $I$ black holes of ${\cal N}=2,\, 4,\, 8$ supergravities. 
\subsection{Ansatz with $m-2$ commuting Killing vectors}
Let us fix our notation for a general $m$-dimensional metric with $m-2$ commuting Killing vectors  \cite{Konoplya:2018arm}) presented as
\begin{equation}
    ds^2 = \G_{ab}dy^a dy^b + \lambda_{ij} dx^i dx^j,
\end{equation}
where $y^a$ are coordinates along the Killing vectors and $i,j=1,2$. Since there are only two coordinates in addition to the Killing directions, one can choose the coordinates $x^i$ such that $x_1 = \text{const}$ gives foliation slices (we use subscripts  for notational appeal). Then the transformation $x_2 \to x_2 - \int dx_1 (\lambda_{12}/\lambda_{22})$ removes the mixed component $\lambda_{12}$ leaving the foliation condition invariant. So, further we will consider the case $\lambda_{11} \equiv \lambda_1$, $\lambda_{22} \equiv \lambda_2$, $\lambda_{12} = 0$, leading  the following expressions
\begin{equation}
    n^\alpha \nabla_\alpha = \lambda_1^{-1/2} \partial_1,\qquad
    \chi = \frac{1}{2} \lambda_1^{-1/2} \partial_1 \ln \lambda_2,\qquad
    \varphi = \lambda_1^{1/2},
\end{equation}
and, from Eq. (\ref{eq:eq_psi}),
\begin{equation}
    \Psi = \ln \lambda_1 + \tilde{\Psi}(x_1),
    \qquad
    \partial_1 \tilde{\Psi}(x_1) = \partial_1 \ln (\lambda_2/\lambda_1).
\end{equation}
The integrability conditions are
\begin{align}
    \partial_2 \partial_1 \ln (\lambda_2/\lambda_1) = 0,\qquad
    \partial_2 \partial_1 (e^{\Psi}\G^{ab}) = 0.
\end{align}
The first condition means that $\lambda_1$ and $\lambda_2$ can be represented as a function of only one variable multiplied by a common factor, depending on both coordinates. Such coordinates can be easily transformed to the form $\lambda_{ij} dx^i dx^j = \lambda(x_1, x_2)(dx_1^2 + dx_2^2)$. In this case the curve $x_1=\text{const}$ is constructed along the \textit{isothermal} coordinate. Then the second condition takes the form $\partial_2 \partial_1 (\lambda \G^{ab}) = 0$ and the metric along the Killing vectors is
\begin{equation}
\G_{ab} = (\mathcal{F}^{-1})_{ab} \lambda,
\end{equation}
where the matrix $\mathcal{F}$ is a matrix of functions of the form $\mathcal{F}^{ab}(x_1, x_2) = \mathcal{X}_1^{ab}(x_1) + \mathcal{X}_2^{ab}(x_2)$. The final structure of the metric is
\begin{equation}
    ds^2 / \lambda = (\mathcal{F}^{-1})_{ab}dy^a dy^b + dx_1^2 + dx_2^2.
\end{equation}
The resulting metric expectedly provides the separability of the massless Hamilton-Jacobi equation. To make the form of the metric less restrictive, one can perform transformations of the form $x'^i = x'^i(x^i)$ (where $i$ is fixed)
\begin{equation} \label{SM}
    ds^2 / \lambda = (\mathcal{F}^{-1})_{ab}dy^a dy^b + f_1(x_1) dx_{1}^2 + f_2(x_2) dx_2^2,
\end{equation}
with the corresponding CKT
\begin{align} \label{SMK}
   & K^{\alpha\beta} =
    \alpha(x_1,x_2) g^{\alpha\beta}
    + \mathcal{X}_1^{ab}(x_1) \delta_{a}{}^\alpha \delta_{b}{}^\beta
    +f_1(x_1)^{-1}\delta^{\alpha}_{x_1} \delta^{\beta}_{x_1},\quad
   \Omega_\alpha =  \partial_\alpha \alpha(x_1,x_2)+ 
   \delta^{x_1}_{\alpha}\cdot \partial_{x_1} \lambda.
\end{align} 
In order to check the integrability conditions for the CKT, one should make sure that the metric can be represented in the form (\ref{SM}). This form of metric and foliation fulfils condition (\ref{eq:block_condition}) automatically, and equation (\ref{eq:umbilic_condition}) is not a condition but an expression for $\chi$ (since there is only one equation with one unknown function $\chi$). The metric form (\ref{SM}) has a formal symmetry $x_1 \leftrightarrow x_2$, therefore if the foliation $x_1 = \text{const}$ generates a CKT, then the foliation $x_2 = \text{const}$ will generate as well.

Additionally, if we are interested in exact KTs, the third integrability condition $\LCS_\gamma (\chi\varphi^3) = 0$ gives $\partial_2 \partial_1 \lambda = 0$, and the factor $\lambda$ must be a function of the form $\lambda = \mathcal{F}^{\lambda}(x_1, x_2) = \mathcal{X}_1^{\lambda}(x_1) + \mathcal{X}_2^{\lambda}(x_2)$ and $\alpha = - \mathcal{X}_1^\lambda(x_1)$ (for comparison see Refs. \cite{benenti,Papadopoulos:2018nvd,Papadopoulos:2020kxu, Carson:2020dez}, where a similar form of the metric was proposed from other considerations).

\subsection{ Plebanski-Demianski solution}
Consider the concrete case of the general Plebansky-Demyansky \cite{Demianski} class of stationary axially symmetric solutions of type $D$ to the Einstein-Maxwell equations with a cosmological constant. The $ds^2$ metric reads from the conformally transformed line element in Boyer-Lindquist coordinates 
\begin{align}  
\Omega^2ds^2&=\Sigma\left(\frac{dr^2}{\Delta_r}+\frac{d\theta^2}{\Delta_\theta}\right)+\frac{1}{\Sigma}
\left((\Sigma+a\chi)^2\Delta_\theta\sin^2\theta-\Delta_r\chi^2\right)d\phi^2 \nonumber \\
\label{Sol222}
&+\frac{2}{\Sigma}\left(\Delta_r\chi-a(\Sigma+a\chi)\Delta_\theta\sin^2\theta\right)dt d\phi-\frac{1}{\Sigma}
\left(\Delta_r-a^2\Delta_\theta\sin^2\theta \right)dt^2, 
\end{align}
where we have defined the following functions
\begin{subequations}
\begin{align}
\Delta_\theta &=1-a_1\cos\theta-a_2\cos^2\theta, \qquad \Delta_r=b_0+b_1r+b_2r^2+b_3r^3+b_4r^4\,,\\
\Omega &=1-q(N+a \cos\theta)r, \quad \Sigma=r^2+(N+a\cos\theta)^2\,,\quad
\chi =a \sin^2\theta-2N(\cos \theta+C)\,,
\end{align}
\end{subequations}
with constant coefficients for $\Delta_\theta$ and $\Delta_r$:
\begin{subequations}
\begin{align}
a_1 &=2aMq-4aN\left(q^2(k+\beta)+\frac{\Lambda}{3}\right), \quad a_2=-a^2\left(q^2(k+\beta)+\frac{\Lambda}{3}\right),\\ 
b_0&=k+\beta, \quad 
b_1 =-2M,\quad
b_2=\frac{k}{a^2-N^2}+4MNq-(a^2+3N^2)\left(q^2(k+\beta)+\frac{\Lambda}{3}\right),\\
b_3 &=-2q\left(\frac{kN}{a^2-N^2}-(a^2-N^2)\left(Mq-N\left(q^2(k+\beta)+\frac{\Lambda}{3}\right)\right)\right),\quad
b_4 =-\left(q^2k+\frac{\Lambda}{3}\right),\\ 
k &=\frac{1+2MN  q-3N^2\left(q^2\beta+\frac{\Lambda}{3}\right)}{1+3q^2N^2(a^2-N^2)}(a^2-N^2), \quad 
q=\frac{A}{\omega}, \quad \omega=\sqrt{a^2+N^2}\,.
\end{align}
\end{subequations}
Seven independent parameters $M,N,a,A,\beta,\Lambda,C$ can be interpreted as the physical charges in the following way: $M, N$ are the mass and the magnetic mass (NUT parameter), $a$ is the Kerr-like rotation parameter, $\beta=e^2+g^2$ comprises the electric $e$ and magnetic $g$ charges, $A$ is the acceleration parameter, $\Lambda$ is the cosmological constant, and the constant $C$ defines the location of the Misner string. 

First, we extract the general conformal factor $\lambda$ from the metric (\ref{Sol222}) such that $g_{rr}=f_r(r)$ and $g_{\theta\theta}=f_\theta(\theta)$ are functions of only the corresponding coordinate: $\lambda=\Sigma /\Omega^2$. Metric (\ref{Sol222}) takes the separable form (\ref{SM}) with
\begin{align}
\mathcal{F}^{ab}=\underbrace{-\Delta^{-1}_r \begin{pmatrix}
     S^2  & aS  \\
      aS  & a^2 \\
\end{pmatrix}}_{\mathcal{X}_r^{ab}(r)} + \underbrace{\Delta^{-1}_\theta \sin^{-2}\theta \begin{pmatrix}
    \chi^2   & \chi  \\
      \chi  & 1 \\
\end{pmatrix}}_{\mathcal{X}_\theta^{ab}(\theta)}, \quad f_r(r) = \Delta^{-1}_r, \quad f_\theta(\theta) = \Delta^{-1}_\theta,
\end{align}
where
\begin{align}
S\equiv\Sigma + a\chi= r^2 + a^2 - 2 a C N + N^2.
\end{align}
Following Eq. (\ref{SMK}), the irreducible CKT reads (which was obtained in Ref. \cite{Kubiznak:2007kh})
 \begin{align} \label{conformalKillingPD}
    &
    K^{\alpha\beta} =
    \alpha g^{\alpha\beta}
    -\Delta^{-1}_r S^\alpha S^\beta
    + \Delta_r \delta_r^\alpha \delta_r^\beta, \quad
    \Omega_\alpha=\partial_\alpha \alpha+ \delta^r_\alpha (\Sigma/\Omega^2)',
    \\\nonumber &
    S^\alpha \equiv (r^2 + a^2 - 2 a C N + N^2) \delta^\alpha_t + a \delta^\alpha_\varphi. 
\end{align}
In the special case of zero acceleration $A=0$ we get $\Omega=1$ and function $\lambda=r^2+(N+a \cos)^2$ is separable, i.e., this CKT can be reduced to the exact KT by choosing $\alpha=-r^2$. 

In accordance with theorem \ref{KBGFPS}, there are always FPSs among slices $r=const$. They form a photon region described by the inequality (\ref{PR1b}) which can be written as \cite{Grenzebach,Grenzebach:2015oea}
 \begin{align} 
4  a^2 \sin^2 \theta \Delta_r \Delta_\theta \Sigma'^2 \geq (2\Delta_r \Sigma' -\Sigma \Delta'_r)^2.
\end{align}
Using general formulas (\ref{eq:general_boundary}) for a static (not asymptotically distant) observer $v^\alpha=\delta^\alpha_t$ we can find a compact expression for the shadow boundary
\begin{subequations}
\begin{align}
 &\cos^2(\Theta_{SH}) =-\frac{\Delta_r -a^2 \sin^2\theta_O \Delta_\theta}{\Delta_r\Sigma^2} \left(-\big((S-\hat{S})+2 \hat{\Delta}_r(\hat{\Sigma}'/\hat{\Delta}'_r)\big)^2+4 \Delta_r \hat{\Delta}_r(\hat{\Sigma}'/\hat{\Delta}'_r)^2\right),\\&
 \sin(\Phi_{SH})\sin(\Theta_{SH}) = \frac{1}{a \Sigma \sqrt{\Delta_r\Delta_\theta} \sin \theta_O}\Big(a^2 \sin^2\theta_O\Delta_\theta (S-\hat{S}) + \Delta_r (\hat{S}-a \chi) \nonumber\\& - 2 \hat{\Delta}_r(\hat{\Sigma}'/\hat{\Delta}'_r )\big(\Delta_r -a^2 \sin^2\theta_O \Delta_\theta\big) \Big),
\end{align}
\end{subequations}
where functions with hats are calculated at FPS, and functions without hats are calculated at the observer's location. From the general Eqs. (\ref{eq:flat_boundary}), a shadow captured by an asymptotic observer (with $w(\theta_O)=a$ similarly to Refs. \cite{Grenzebach,Grenzebach:2015oea,GrenzebachSBH,Perlick:2021aok}) in asymptotically flat spacetime ($A=0$, $\Lambda=0$) has the following boundary
\begin{subequations}
\begin{align}
& R^2_{SH} =   \frac{4r^2_{PS}(r^2_{PS}-2 M r_{PS}+a^2+\beta - N^2)}{(r_{PS}-M)^2},\\
&\rho_{\mathcal N} =-\frac{r^3_{PS}-3 M r^2_{PS}+(a^2-3 N^2+2 \beta)r_{PS}+M(a^2 +N^2)}{a(r_{PS}-M)}+2 N \cos \theta_O- a \sin^2\theta_O.
\end{align}
\end{subequations}

\subsection{NUT wormholes}
Metrics endowed with the Newman, Tamburino and Unti (NUT) parameter \cite{Newman:1963yy} may describe wormholes \cite{Clement:2015aka} violating the energy conditions only in distributional sense. Here we discuss the non-rotating NUT wormhole as an example of a metric with more than two irreducible Killing vectors, from which one can choose different commuting pairs.  The line element depending with mass $M$, NUT charge $N$, combined electric-magnetic charge $e$ and the Misner string \cite{Misner:1963fr} parameter $C$  reads
\begin{equation} \label{SolNUT}
    ds^2 =
    - \frac{\Delta }{\Sigma} \left(dt - \omega d\varphi\right)^2
    + \Sigma \left(
        \frac{dr^2}{\Delta} + d\theta^2 + \sin^2\theta d\varphi^2
    \right),
\end{equation}
where functions $\Delta$, $\omega$, $\Sigma$ are defined as follows
\begin{align}
   & \Delta = r^2-2M r-N^2+ e^2, \quad \Sigma = r^2 + N^2, \quad \omega =- 2 N (\cos \theta + C).
\end{align}
This metric has four Killing vectors
\begin{subequations}
\begin{align}
  & K_{(t)}^\alpha=\delta^\alpha_t,\\
  & K_{(x)}^\alpha=\frac{2N(1+C \cos \theta) \cos\phi}{\sin\theta}\delta^\alpha_t-\sin \phi \delta^\alpha_\theta-\cos \phi \cot \theta \delta^\alpha_\phi,\\
  & K_{(y)}^\alpha=\frac{2N(1+C \cos \theta) \sin\phi}{\sin\theta}\delta^\alpha_t+\cos \phi \delta^\alpha_\theta-\sin \phi \cot \theta \delta^\alpha_\phi,\\
  & K_{(z)}^\alpha=\delta^\alpha_\phi -2 N C \delta^\alpha_t.
\end{align}
\end{subequations}
All of them are tangent to the slices $r=\text{const}$ and constitute $\mathfrak{s}\mathfrak{o}(3)$ algebra. The polar axes contain the singularity of the Misner string, which can be removed in the north or south pole, choosing $C=\pm 1$. If time is a periodic coordinate $\Delta t = 8\pi N$, the Misner string can be removed from the spacetime \cite{Misner:1963fr}, turning the slice $r=\text{const}$ into $S^3$ in terms of a Hopf fibration.

For the foliation $r=\text{const}$ one can get $n^\alpha=\sqrt{\Delta/\Sigma}\delta^\alpha_r$ and $\varphi = \sqrt{\Sigma/\Delta}$.
The vector orthogonal to the chosen pair of Killing vectors will be denoted as $\X^\alpha$. For pairs of the form $(K_{(t)}^\alpha,K_{(i)}^\alpha)$ (where $i=x,\,y,\,z$) we find
\begin{align} \label{NUTchi}
 \chi=\frac{1}{2}\sqrt{\Delta/\Sigma} (\ln \Sigma)',  \quad \SFS_{\alpha\beta}K_{(t)}^\alpha \X^\beta=\SFS_{\alpha\beta}K_{(i)}^\alpha \X^\beta=0,
 \end{align}
and for pairs of the form $(K_{(i)}^\alpha,K_{(j)}^\alpha)$ with $i \neq j$ the off-diagonal block reads
\begin{align}
 \SFS_{\alpha\beta}K_{(i)}^\alpha \X^\beta=-N Q_{(i)}(r,\theta,\phi) (\Sigma \Delta'-2 \Delta \Sigma'),
\end{align}
where $Q_{(i)}$ is some function. The block diagonal condition (\ref{eq:block_condition}) is satisfied identically only for pairs of the form $(K_{(t)}^\alpha,K_{(i)}^\alpha)$. In the case $N\neq 0$, other pairs fulfil condition (\ref{eq:block_condition}) only at one slice distinguished by the condition $\Sigma \Delta'-2 \Delta \Sigma'=0$, i.e., at the photon surface \cite{Grenzebach,Grenzebach:2015oea}. Pairs without $K_{(t)}^\alpha$ will be not considered further. Using the expression (\ref{NUTchi}) for $\chi$, one can find that the condition (\ref{eq:integrability_G_a}) is satisfied and $\Psi= \ln \Sigma$. The second integrability condition in form  (\ref{eq:integrability_B}) is also satisfied for such pairs, and we obtain
\begin{align}
\mathcal{F}^{ab}=\underbrace{-\Delta^{-1} \begin{pmatrix}
     \Sigma^2  & 0  \\
      0 & 0 \\
\end{pmatrix}}_{\mathcal{X}_r^{ab}(r)} + \underbrace{H^{-1}_{(i)}(\theta,\phi)\begin{pmatrix}
    N^2P^2_{(i)}(\theta,\phi)   & NP_{(i)}(\theta,\phi)  \\
     NP_{(i)}(\theta,\phi)  & 1 \\
\end{pmatrix}}_{\mathcal{X}_{(i)}^{ab}(\theta,\phi)},
\end{align}
where  $H_{(i)}(\theta,\phi)$ and  $P_{(i)}(\theta,\phi)$ are some functions of the corresponding coordinates possessing the symmetry along the Killing vector $K_{(i)}^\alpha$. All three pairs gives the same irreducible CKT
\begin{align}
    K^{\alpha\beta} =
    \alpha  g^{\alpha\beta}
    -\Delta^{-1} S^\alpha S^\beta
    + \Delta \delta_r^\alpha \delta_r^\beta, \quad \Omega_\alpha=\partial_\alpha \alpha+ 2r \delta^r_\alpha,\quad
    S^\alpha = \left(r^2 + N^2\right) \delta^\alpha_t,
\end{align}
which is a special case of (\ref{conformalKillingPD}). The result with the same CKT from three different pairs is expected due to (local) spherical symmetry.
\subsection{EMDA black holes}
The stationary charged Einstein-Maxwell-dilaton-axion (EMDA) black hole solutions with NUT, relevant to ${\cal N}=4$ supergravity or the heterotic string theory,  depending  on seven parameters: mass $M$,  electric and magnetic charges $Q,P$, rotation parameter $a$,   NUT $N$ and asymptotic values of the dilaton and axion fields (irrelevant for the metric) was obtained in \cite{Galtsov:1994pd}. Black hole solution without NUT was previously derived by A. Sen \cite{Sen:1992ua} and now is commonly referred as Kerr-Sen metric. Non-rotating solutions with NUT were independently obtained by Kallosh et al. \cite{Kallosh:1994ba}  and Johnson and Myers \cite{Johnson:1994nj}. The Kerr-Sen metric is often considered as a deformed Kerr in modelling deviations from the standard picture of black holes \cite{Zhang:2020pay}. The line element of the solution can be written in the Kerr-like form in Boyer-Lindquist coordinates
\begin{equation} \label{SolGK}
    ds^2 =
    - \frac{\Delta - a^2 \sin^2\theta}{\Sigma} \left(dt - \omega d\varphi\right)^2
    + \Sigma \left(
        \frac{dr^2}{\Delta} + d\theta^2 + \frac{\Delta \sin^2\theta}{\Delta - a^2 \sin^2\theta} d\phi^2
    \right),
\end{equation}
where the functions $\Delta$, $\omega$, $\Sigma$ are redefined as follows
\begin{subequations}
\begin{align}
   & \Delta = (r - r_{-}) (r - 2M) + a^2 - (N-N_{-})^2, \\
   & \Sigma = r(r-r_{-}) + (a\cos\theta + N)^2 - N_{-}^2, \\
    &\omega = \frac{-2W}{\Delta - a^2 \sin^2\theta},\quad
     W =  N \Delta \cos\theta
        + a \sin^2\theta \left( M(r-r_{-}) + N(N - N_{-}) \right)
\end{align}
\end{subequations}
with abbreviations for constants
\begin{equation}
    r_{-} = \frac{M |\mathcal{Q}|^2 }{|\mathcal{M}|^2},\quad
    N_{-} = \frac{N |\mathcal{Q}|^2 }{2|\mathcal{M}|^2},\quad
    \mathcal{M} = M + i N,\quad
    \mathcal{Q} = Q - i P.
\end{equation}
The full solution also contains Maxwell, dilaton $\phi$ and axion $\kappa$ fields (whose form can be found in \cite{Galtsov:1994pd}). The solution reduces to Kerr-NUT for $Q=P=0$. It has been shown to belong to the general Petrov type $I$ \cite{Garcia:1995qz,Burinskii:1995hk}, in contrast to the similar solution in Einstein-Maxwell model, Kerr-Newman-NUT, which belongs to the type $D$.

To construct the CKT, first of all, similarly to the previous example, we extract the conformal factor $\lambda$ from the metric $\lambda=\Sigma$.
After that, metric (\ref{SolGK}) takes the separable form (\ref{SM}) with
\begin{align}
\mathcal{F}^{ab}=\underbrace{-\Delta^{-1} \begin{pmatrix}
     S^2  & aS  \\
      aS  & a^2 \\
\end{pmatrix}}_{\mathcal{X}_r^{ab}(r)} + \underbrace{ \sin^{-2}\theta \begin{pmatrix}
    \chi ^2   & \chi  \\
      \chi   & 1 \\
\end{pmatrix}}_{\mathcal{X}_\theta^{ab}(\theta)}, \quad f_r(r) = \Delta^{-1}, \quad f_\theta(\theta) = 1,
\end{align}
where
\begin{align}
&\chi \equiv  a \sin^2 \theta - 2 N \cos \theta,  \quad S\equiv r(r-r_{-}) + a^2 + N^2 - N_{-}^2.
\end{align}
The irreducible CKT (\ref{SMK}) is
\begin{align}
    &
    K^{\alpha\beta} =
    \alpha  g^{\alpha\beta}
    -\Delta^{-1} S^\alpha S^\beta
    + \Delta \delta_r^\alpha\delta_r^\beta, \quad \Omega_\alpha=\partial_\alpha \alpha+ (2r-r_-) \delta^r_\alpha,
    \\\nonumber &
    S^\alpha = \left(r(r-r_{-}) + a^2 + N^2 - N_{-}^2\right) \delta^\alpha_t + a \delta^\alpha_\varphi.
\end{align}
Since $\lambda$ is separable, we can always get the KT by choosing $\alpha=-r(r-r_{-})$ \cite{Kobialko:2021qat}. For Kerr-Sen solution see also Ref. \cite{Houri:2010fr}.

\subsection{Rotating asymptotically flat STU black holes}
STU black holes represent general ${\cal N}=8$ supergravity solutions \cite{Chow:2014cca}.  In this subsection we focus on rotating electrically charged black holes in ungauged STU supergravity, which are characterised by the mass $M$, the angular momentum $J$ and four electric charges parameterized by constants $s_i, i=1,2,3,4$. Employing solution-generating techniques, this family of metrics was first obtained in \cite{Cvetic:1996kv} (see also Ref. \cite{Chong:2004na,Cvetic:2017zde}) and reads
\begin{equation} \label{SolSTU}
    ds^2 =
    - \frac{\Delta - a^2 \sin^2\theta}{\Sigma} \left(dt - \omega d\varphi\right)^2
    + \Sigma \left(
        \frac{dr^2}{\Delta} + d\theta^2 + \frac{\Delta \sin^2\theta}{\Delta - a^2 \sin^2\theta} d\phi^2
    \right),
\end{equation}
where functions $\Delta$, $\omega$, $\Sigma$ are defined as follows
\begin{subequations}
\begin{align}
   &
    \Delta = r^2 -2 M r+ a^2, \quad
    \omega = -\frac{2M a W \sin^2\theta}{\Delta - a^2 \sin^2\theta}, \quad
    W = (\Pi_c-\Pi_s)r+2 M \Pi_s,
    \\ &
    \Sigma^2 =
    \prod^4_i (r+ 2 M s^2_i)+a^4 \cos^4 \theta
    +2 a^2 R^2 \cos^2 \theta,
    \\ &
    R^2 = r^2+ M r \sum^4_i s^2_i +4 M^2(\Pi_c-\Pi_s)\Pi_s-2 M^2 \sum^4_{i<j<k}s^2_is^2_js^2_k,
\end{align}
\end{subequations}
with constants 
\begin{align}
\Pi_s= \prod^4_i s_i, \quad \Pi_c = \prod^4_i \sqrt{1+s^2_i}.
\end{align}
 
Similarly, we extract the general conformal factor $\lambda$ from the metric $\lambda=\Sigma$ and find the separable form of the metric (\ref{SolSTU})
\begin{align}
\mathcal{F}^{ab}=\underbrace{-\Delta^{-1} \begin{pmatrix}
     S^2  & aS  \\
      aS  & a^2 \\
\end{pmatrix}-\begin{pmatrix}
     Q  & 0  \\
      0  & 0 \\
\end{pmatrix}}_{\mathcal{X}_r^{ab}(r)} + \underbrace{ \begin{pmatrix}
    -a^2 \cos^{2}\theta & 0  \\
      0   & \sin^{-2}\theta \\
\end{pmatrix}}_{\mathcal{X}_\theta^{ab}(\theta)}, \quad f_r(r) = \Delta^{-1},
\end{align}
where
\begin{align}
&Q \equiv 2R^2 - r(r - 2M), \quad S \equiv 2 M W.
\end{align}
The irreducible CKT (\ref{SMK}) is
\begin{align}
    &
    K^{\alpha\beta} =
    \alpha  g^{\alpha\beta}
    -\Delta^{-1} S^\alpha S^\beta-\left( r^2 +2 Mr \left(1+  \sum^4_i s^2_i\right)\right)\delta_t^\alpha\delta_t^\beta
    + \Delta \delta_r^\alpha \delta_r^\beta, \\ &\nonumber \Omega_\alpha=\partial_\alpha \alpha+ \Sigma' \delta^r_\alpha,
    \quad
    S^\alpha = 2 M \left((\Pi_c-\Pi_s)r+2 M \Pi_s\right) \delta^\alpha_t + a \delta^\alpha_\varphi.
\end{align}
The CKT can be reduced to the KT only if $\Sigma$ is separable. The mixed derivative $\partial_r\partial_\theta\Sigma$ has the form of the fraction with polynomial numerator. Zeroing the coefficients at all powers $r^a \cdot \cos^b\theta$ of the numerator, we find that $\Sigma$ is separable if and only if all the following conditions are satisfied (for $M\neq0$ and $a\neq0$)
\begin{align}
s_i^2=s_j^2, \quad s_k^2 = s_l^2, \quad s_i s_k = s_j s_l,
\end{align}
where $(i, j, k, l)$ is any permutation of $(1,2,3,4)$. Consider a few special cases \cite{Cvetic:2017zde}. 
\begin{enumerate}
    \item Einstein-Maxwell black holes with $s_i=s$. Conformal factor $$\lambda=(r+2 M s^2)^2 +a^2 \cos^2 \theta$$ is separable, and there is a KT with $ \alpha =-(r+2 M s^2)^2$.
    \item Kerr-Sen black holes with $s_1=s_3=s$, $s_2=s_4=0$. Conformal factor  $$\lambda=\Sigma=r(r+2 M s^2) +a^2 \cos^2 \theta$$ is separable, and there is a KT with $\alpha =-r(r+2 M s^2)$.
    \item Black holes with pairwise-equal charges $s_1=s_3=s$, $s_2=s_4=s'$. Conformal factor $$\lambda=\Sigma=(r+2 M s^2)(r+2 M s'^2) +a^2 \cos^2 \theta$$ is separable, and there is a KT with $\alpha =-(r+2 M s^2)(r+2 M s'^2)$. 
    \item Rotating Kaluza-Klein charged black holes with $s_1=s$, $s_2=s_3=s_4=0$. Conformal factor $$\lambda=\Sigma=\sqrt{(r^2 +a^2 \cos^2 \theta)(r^2+2 Ms^2r+a^2 \cos^2 \theta)}$$ is not separable, and there is only a CKT.
\end{enumerate}

Following Eq. (\ref{eq:flat_boundary}) the boundary of the STU black hole shadow captured by an asymptotically distant static observer ($w(\theta_O)=0$) is governed by the following expression
\begin{subequations}
\begin{align}
& R^2_{SH}= \frac{(a \rho_{PS} -2 M W_{PS} )^2}{\Delta_{PS}}-\frac{4 M^2 W^2_{PS}-\Sigma^2(r_{PS},\theta_O)}{\Delta_{PS}-a^2 \sin^2\theta_O}, \quad \rho_{\mathcal N}=\rho_{PS},\\
&\rho_{PS}=\frac{1}{a(M-r_{PS})}\Big( M \big((r_{PS}-2M)^2 \Pi_s -r^2_{PS}\Pi_c\big)+a^2 M (\Pi_c-\Pi_s) \\ &\pm \Delta_{PS} \big(r^2_{PS}+M r_{PS}\sum^4_i s^2_i +M^2 (2\Pi_s(\Pi_s-\Pi_c)+\sum^4_{i<j}s^2_is^2_j +\sum^4_{i<j<k}s^2_is^2_js^2_k  ) \big)^{1/2}\Big).\nonumber 
\end{align}
\end{subequations}

\subsection{General black holes in $D=4,\;\mathcal{N}=2$ supergravity}
The most general asymptotically flat stationary non-extremal dyonic black hole in $D=4$ $\mathcal{N}=2$ supergravity coupled to 3 vector multiplets was presented in Ref. \cite{Chow:2013tia} (for more details see Ref. \cite{Chow:2014cca}). It describes the low-energy solution of the STU model. This family of solutions depends on 11 independent parameters: mass $M$, NUT $N$, rotation parameter $a$, four electric $Q_i$ and four magnetic $P_i$ charges. The line element reads
\begin{equation} \label{SolGSTU}
    ds^2 =
    - \frac{\Delta_r - \Delta_u}{\Sigma} \left(dt - \omega d\varphi\right)^2
    + \Sigma \left(
        \frac{dr^2}{\Delta_r} + \frac{du^2}{\Delta_u} + \frac{\Delta_r \Delta_u}{a^2(\Delta_r - \Delta_u)} d\phi^2
    \right),
\end{equation}
where the functions $\Delta_r$, $\Delta_u$, $\omega$, $\Sigma$ are defined as follows
\begin{subequations}
\begin{align}
   & \Delta_r = r^2 -2 M r+ a^2-N^2, \quad \Delta_u=a^2-(u-N)^2, \\ 
   & \omega = -\frac{2 N_{ph}(u-N)\Delta_r +(W_r+2 N_{ph} N)\Delta_u}{a(\Delta_r - \Delta_u)}, \\
   & \Sigma^2=(\Delta_r - \Delta_u)^2+2(\Delta_r - \Delta_u)(2 M r+W_u)+(2 N_{ph} u + W_r)^2,
\end{align}
\end{subequations}
where $W_r=a_r r+b_r$ and $W_u=a_u u+b_u$ are linear functions of the corresponding coordinates. The physical charges correspond to these parameters in a complicated way. In particular, physical mass and NUT charge are
\begin{align}
  M_{ph}=\mu_1 M +\mu_2 N \quad N_{ph}=\nu_1 M +\nu_2N.
\end{align}
The complete description of the solution is given in Ref. \cite{Chow:2013tia}.

The conformal factor to be extracted from the metric (\ref{SolGSTU}) for CKT generation is $\lambda=\Sigma$. After that, the metric (\ref{SolGSTU}) takes the separable form (\ref{SM}) with
\begin{align}
&\mathcal{F}^{ab}=\underbrace{-\Delta^{-1}_r \begin{pmatrix}
     S^2  & aS  \\
      aS  & a^2 \\
\end{pmatrix}-\begin{pmatrix}
     S'  & 0  \\
      0 & 0 \\
\end{pmatrix}}_{\mathcal{X}_r^{ab}(r)} + \underbrace{ \Delta^{-1}_u \begin{pmatrix}
    Q^2 & a Q  \\
      a Q   & a^2 \\
\end{pmatrix}+\begin{pmatrix}
    Q' & 0  \\
      0   &  0 \\
\end{pmatrix}}_{\mathcal{X}_u^{ab}(u)}, \\ \nonumber
&f_r(r) = \Delta^{-1}_r, \quad f_u(u) = \Delta^{-1}_u
\end{align}
where
\begin{subequations}
\begin{align}
&S \equiv W_r +2 N_{ph} N, \quad S'= \Delta_r+4 M r,\\
&Q \equiv 2 N_{ph}(N-u), \quad Q'= \Delta_u -2 W_u.
\end{align}
\end{subequations}
The corresponding irreducible CKT (\ref{SMK}) reads
\begin{align}
    &
    K^{\alpha\beta} =
    \alpha  g^{\alpha\beta}
    -\Delta^{-1}_r S^\alpha S^\beta-\left(\Delta_r+4 M r\right)\delta_t^\alpha \delta_t^\beta
    + \Delta_r \delta_r^\alpha \delta_r^\beta, \\ & \Omega_\alpha=\partial_\alpha\alpha+ \Sigma' \delta^r_\alpha,
    \quad
    S^\alpha = \left(W_r +2 N_{ph} N\right) \delta^\alpha_t + a \delta^\alpha_\varphi.
\end{align}
The CKT can be reduced to an exact KT if and only if $\Sigma$ is separable. Computing the mixed derivative $\partial_r\partial_u\Sigma$ and zeroing the coefficients at all powers $u^a r^b$, one finds that $\Sigma$ is separable if
\begin{align}
a_r=\pm2 M, \quad a_u= \pm2 N_{ph}, \quad b_r=\pm b_u.
\end{align}
The form of the conformal factor becomes separable
\begin{align}
\lambda=\Sigma=r^2+ b_u + u (u-2 N\pm 2 N_{ph}),
\end{align}
and the exact KT exists with $\alpha =-r^2$.

\section{Conclusions}
We have presented a purely geometric method of generating a CKT of the second rank in spacetimes with foliation of codimension one and additional isometries. 
For spacetimes foliated into arbitrary slices, we have obtained the general  lift equations  (\ref{eq:killing_map_xy}) relating  CKT to its slice projections  and the orthogonal complement. Using these equations one can try to raise a reducible CKT, defined in such slices, to the CKT in the bulk that may turn out to be irreducible. Integrability conditions (\ref{eq:integrability_G}) that ensures a successful lifting for the given foliation were established. The resulting generation method was presented in the theorem \ref{KBG}.  The ``slice-reducible'' CKT obtained with the help of this algorithm, being projected onto the corresponding slices, is reducible in them. Apart from the projection, the full tensor has a normal component but no mixed components. We show how this construction is reduced to the similar construction \ref{KBG_non} for exact KTs, providing a more elegant and clear statement of the previous results \cite{Kobialko:2021aqg}. 
 
For the particular case of totally umbilic foliations, the resulting CKT is reducible in the bulk. However, even in this case, the integrability conditions guarantee the existence of a new conformal Killing vector (\ref{eq:conformal_reducible}) that is normal to slices and provides a new independent integral of motion. In particular, a geometry with $n=m-2$ conformal Killing vectors and the second fundamental form (\ref{eq:block_condition}), (\ref{eq:umbilic_condition}) that satisfies the integrability conditions (\ref{eq:integrability_G}) always corresponds to a completely integrable dynamical system for null geodesics. It is worth noting that in the case of totally umbilical slices, slice-reducible CKTs with non-zero mixed components $\beta^\alpha$, which may turn out to be irreducible, can also arise. Although for the completely geodesic foliation considered in Ref. \cite{Garfinkle:2010er} no examples of irreducible KT have been obtained. 

Finding a foliation that satisfies the integrability conditions can be a difficult task. However, we succeeded in proving (theorem \ref{KBGFPS}) that slices satisfying the integrability conditions must contain fundamental photon  surfaces if the photon region inequalities hold. 
Thus, the presence of an FPS is a necessary condition that does not yet guarantee the existence of a CKT, but serves as an indication that such a tensor can exist. 
Therefore, it is recommended to check the integrability conditions using the fundamental photon surfaces as slices in the CKT generating technique. This generalizes and clarifies the relationship between FPS and CKT discussed earlier in Refs. \cite{Koga:2020akc,Kobialko:2021aqg,Pappas:2018opz}. It is tempting to conjecture that the existence of fundamental photon surfaces implies the existence of CKT if the slices are equipotential (ADM lapse function $N= \text{const}$ in each slice) or spherical \cite{Cederbaum:2019rbv} or the photon region is one-sheeted in the sense of Ref. \cite{Kobialko:2020vqf}.   
 
Since a geometry with $n=m-2$ conformal Killing vectors that satisfies the integrability conditions (\ref{eq:integrability_G})  corresponds to an integrable dynamical system for null geodesics, one can obtain a fully analytical description of the gravitational shadows of the corresponding ultra-compact objects. 
We have found a general analytical expression (\ref{eq:general_boundary}) for the shadow boundary in an arbitrary geometry with a slice-reducible CKT and its limit for a distant observer. These equations provide a simple basis for shadow analysis, requiring only knowledge of two conformal Killing vectors, slice-reducible CKT (irreducible in the bulk), and the structure of fundamental photon surfaces associated with this CKT. The developed formalism is easily applicable in practical calculations and, at a new level, reveals a deep connection between photon surfaces and the optical properties of a gravitating object. 

We have applied the developed methods to the Plebansky-Demyansky, EMDA, STU and some other solutions where slice-reducible CKTs arise purely algebraically without solving any differential equations. We analyzed whether CKT can be reduced to exact KT in these examples and gave analytical expressions for the boundaries of the gravitational shadows. We also considered an example of a metric with an excess number of Killing vectors and showed that all possible pairs suitable for generation, as expected, generate the same irreducible CKT. 

\begin{acknowledgments}
The work was supported by the Russian Foundation for Basic Research on the project 20-52-18012 Bulg-a, and the Scientific and Educational School of Moscow State University “Fundamental and Applied Space Research”. I.B. is also grateful to the Foundation for the Advancement of Theoretical Physics and Mathematics ``BASIS'' for support.
\end{acknowledgments}

\appendix

\section{Proof for commutator identities}
\label{auxiliary_identities}

First, the commutator $[n^\alpha \nabla_\alpha, \LCS_\gamma]$ acting on a scalar function is
\begin{align}
    [n^\alpha \nabla_\alpha, \LCS_\gamma] \omega & = 
    [n^\alpha \nabla_\alpha, h^\beta_\gamma \nabla_\beta] \omega = 
      n^\alpha (\nabla_\alpha h^\beta_\gamma) \nabla_\beta \omega
    - h^\beta_\gamma (\nabla_\beta n^\alpha) \nabla_\alpha \omega
     \\\nonumber & =
      n^\alpha \left(
        - \epsilon \SFS_{\alpha}^{\beta} n_\gamma
        - \epsilon n^\beta \LCM_\alpha n_\gamma
        - n_\alpha n_\gamma n^\lambda \LCM_{\lambda} n^{\beta}
      \right) \nabla_\beta \omega
    - h^\beta_\gamma \left(
          \SFS^{\alpha}_{\beta}
        + \epsilon n_\beta n^\lambda \LCM_{\lambda} n^\alpha
    \right) \nabla_\alpha \omega
     \\\nonumber & =
    - \epsilon \left(
          n^\beta n^\lambda \LCM_{\lambda} n_\gamma
        + n_\gamma n^\lambda \LCM_{\lambda} n^{\beta}
      \right) \nabla_\beta \omega
    - \SFS^{\alpha}_{\gamma}\nabla_\alpha \omega
     \\\nonumber & =
    - \epsilon n^\lambda \LCM_{\lambda} n_\alpha \left(
          n^\beta h_\gamma^\alpha 
        + n_\gamma h^{\beta\alpha}
      \right) \nabla_\beta \omega
    - \SFS^{\alpha}_{\gamma}\nabla_\alpha \omega
     \\\nonumber & =
    \LCS_\alpha \ln \varphi \left(
          n^\beta h_\gamma^\alpha 
        + n_\gamma h^{\beta\alpha}
      \right) \nabla_\beta \omega
    - \SFS^{\alpha}_{\gamma}\nabla_\alpha \omega.
\end{align}
 Using this, let us find the expression for the following commutator
\begin{align}
    [ \varphi n^\alpha \nabla_\alpha, \LCS_\gamma] \omega &=
      \varphi n^\alpha \nabla_\alpha \LCS_\gamma \omega
    - \LCS_\gamma( \varphi n^\alpha \nabla_\alpha \omega) =
    \varphi \left(
          [n^\alpha \nabla_\alpha, \LCS_\gamma] \omega
        - \LCS_\gamma \ln \varphi \cdot n^\alpha \nabla_\alpha \omega
    \right)
     \\\nonumber & =
    \varphi \left(
        \LCS_\alpha \ln \varphi \left(
              n^\beta h_\gamma^\alpha 
            + n_\gamma h^{\beta\alpha}
        \right) \nabla_\beta \omega
        - \SFS^{\alpha}_{\gamma}\nabla_\alpha \omega
        - \LCS_\gamma \ln \varphi \cdot n^\alpha \nabla_\alpha \omega
    \right)
     \\\nonumber & =
    \left(
        n_\gamma \LCS^\alpha \varphi
        - \varphi \SFS^{\alpha}_{\gamma}
    \right)\LCS_\alpha \omega.
\end{align}

\section{Proof of proposition \ref{prop:projected_killing_equations}}
\label{proof_proposition_1}

Substitute the split form of the vector $K^\alpha=\K^{\alpha} + \zeta n^{\alpha}$ into the left hand side of Eq. (\ref{eq:killing_equation_a}) and act with different projectors 
\begin{subequations} 
\begin{align}
h^{\alpha}_\rho h^{\beta}_\lambda \LCM_{(\alpha} K_{\beta)}&=h^{\alpha}_{(\rho} h^{\beta}_{\lambda)} (\LCM_{\alpha}\K_{\beta} + \LCM_{\alpha}(\zeta)n_{\beta}+\zeta\LCM_{\alpha} n_{\beta})\\&=\LCS_{(\rho}\K_{\lambda)}+\zeta h^{\alpha}_{(\rho} h^{\beta}_{\lambda)}\LCM_{\alpha} n_{\beta}=\LCS_{(\rho}\K_{\lambda)} + 2 \zeta \SFS_{\rho\lambda},\nonumber
\end{align}
\begin{align}
h^{\alpha}_\rho n^{\beta} \LCM_{(\alpha} K_{\beta)}&=h^{\alpha}_{\rho} n^{\beta} (\LCM_{(\alpha}\K_{\beta)} + \LCM_{(\alpha}(\zeta)n_{\beta)}+\zeta\LCM_{(\alpha} n_{\beta)})\\&=h^{\alpha}_{\rho} n^{\beta}\LCM_{(\alpha}\K_{\beta)} + h^{\alpha}_{\rho} n^{\beta}\LCM_{(\alpha}(\zeta)n_{\beta)}+h^{\alpha}_{\rho} n^{\beta}\zeta\LCM_{(\alpha} n_{\beta)}\nonumber\\&=-\SFS^{\beta}_{\rho}\K_{\beta}+h^{\alpha}_{\rho} n^{\beta}\LCM_{\beta}\K_{\alpha} + \epsilon  \LCS_{\rho}(\zeta)-\epsilon \zeta  \LCS_{\rho} \ln \varphi,\nonumber
\end{align}
\begin{align}
n^{\alpha} n^{\beta} \LCM_{(\alpha} K_{\beta)}&=n^{\alpha} n^{\beta} (\LCM_{(\alpha}\K_{\beta)} + \LCM_{(\alpha}(\zeta)n_{\beta)}+\zeta\LCM_{(\alpha} n_{\beta)})\\&=2 n^{\alpha} n^{\beta}\LCM_{\alpha}\K_{\beta} +2\epsilon n^{\alpha}\LCM_{\alpha}(\zeta).\nonumber 
\end{align}
\end{subequations}
The projections of the right hand side of Eq. (\ref{eq:killing_equation_a}) follows from the definition of $h_{\alpha\beta}$.

\section{Proof of proposition \ref{prop:killing_map}}
\label{proof_proposition_2}
Similarly to the previous proof, substitute the split form of the tensor $K_{\alpha\beta}=\K_{\alpha\beta} + \zeta n_{\alpha} n_{\beta} + \beta_{(\alpha}n_{\beta)}$ into the left hand side of Eq. (\ref{eq:killing_equation_b}) and act with different projectors
\begin{subequations} 
\begin{align}
    h^\alpha_\rho h^\beta_\sigma h^\gamma_\tau
    \LCM_{(\alpha} K_{\beta\gamma)} & = 
    h^\alpha_\rho h^\beta_\sigma h^\gamma_\tau
    \LCM_{(\alpha} \left(
      \K_{\beta\gamma)}
    + \zeta n_{\beta} n_{\gamma)}
    + 2\beta_{\beta}n_{\gamma)}
    \right)
     \\\nonumber & = 
    h^\alpha_\rho h^\beta_\sigma h^\gamma_\tau
    \left(
      \LCM_{(\alpha} \K_{\beta\gamma)}
    + \LCM_{(\alpha} (\zeta n_{\beta} n_{\gamma)})
    + 2\LCM_{(\alpha}(\beta_{\beta}n_{\gamma)})
    \right)
     \\\nonumber  & = 
    h^\alpha_\rho h^\beta_\sigma h^\gamma_\tau
    \left(
      \LCM_{(\alpha} \K_{\beta\gamma)}
    + 2\beta_{(\alpha}\LCM_{\beta}n_{\gamma)}
    \right)
     \\\nonumber  & = 
      \LCS_{(\rho} \K_{\sigma\tau)}
    + 2 \beta_{(\rho} \SFS_{\sigma}{}_{\tau)},
\end{align}
\begin{align}
    n^\alpha h^\beta_\sigma h^\gamma_\tau
    \LCM_{(\alpha} K_{\beta\gamma)} &  = 
    n^\alpha h^\beta_\sigma h^\gamma_\tau
    \LCM_{(\alpha} \left(
      \K_{\beta\gamma)}
    + \zeta n_{\beta} n_{\gamma)}
    + 2\beta_{\beta}n_{\gamma)}
    \right)
    \\\nonumber  & = 
    n^\alpha h^\beta_\sigma h^\gamma_\tau
    \left(
      \LCM_{(\alpha} \K_{\beta\gamma)}
    + \LCM_{(\alpha} (\zeta n_{\beta} n_{\gamma)})
    + 2\LCM_{(\alpha}(\beta_{\beta}n_{\gamma)})
    \right)
    \\\nonumber  &= 
      n^\alpha h^\beta_\sigma h^\gamma_\tau
      \LCM_{(\alpha} \K_{\beta\gamma)}
    + 2 n^\alpha h^\beta_\sigma h^\gamma_\tau\left(\zeta 
      n_{(\alpha} \LCM_{\beta} n_{\gamma)}
    +   
      n_{(\alpha} \LCM_{\beta}\beta_{\gamma)}
    + 
      \beta_{(\alpha} \LCM_{\beta}n_{\gamma)}\right)
     \\\nonumber & = 
      n^\alpha h^\beta_\sigma h^\gamma_\tau
      \LCM_{(\alpha} \K_{\beta\gamma)}
    + 2 \epsilon \zeta h^\beta_\sigma h^\gamma_\tau
      \LCM_{(\beta} n_{\gamma)}
    + 2 \epsilon h^\beta_\sigma h^\gamma_\tau 
      \LCM_{(\beta}\beta_{\gamma)}
    + 2  h^\beta_\sigma h^\gamma_\tau
      (n^\alpha\LCM_{\alpha}n_{(\beta}) \beta_{\gamma)}
    \\\nonumber & = 
      n^\alpha h^\beta_\sigma h^\gamma_\tau
      \LCM_{(\alpha} \K_{\beta\gamma)}
    + 2 \epsilon \zeta \chi_{(\sigma\tau)}
    + 2 \epsilon \LCS_{(\sigma}\beta_{\tau)}
    - 2 \epsilon \beta_{(\sigma} \LCS_{\tau)} \ln \varphi  \\\nonumber  & = 
      h^\beta_{(\sigma} h^\gamma_{\tau)}
      n^\alpha \LCM_{\alpha} \K_{\beta\gamma}- 2\SFS^\alpha_{(\sigma} 
      \K_{\tau)\alpha}
    + 2 \epsilon \zeta \chi_{(\sigma\tau)}
    + 2 \epsilon \LCS_{(\sigma}\beta_{\tau)}
    - 2 \epsilon \beta_{(\sigma} \LCS_{\tau)} \ln \varphi,
\end{align}
\begin{align}
    n^\alpha n^\beta h^\gamma_\tau
    \LCM_{(\alpha} K_{\beta\gamma)}   & = 
    n^\alpha n^\beta h^\gamma_\tau
    \LCM_{(\alpha} \left(
      \K_{\beta\gamma)}
    + \zeta n_{\beta} n_{\gamma)}
    + 2\beta_{\beta}n_{\gamma)}
    \right)
     \\\nonumber & = 
    n^\alpha n^\beta h^\gamma_\tau
    \left(
      \LCM_{(\alpha} \K_{\beta\gamma)}
    + \LCM_{(\alpha} (\zeta n_{\beta} n_{\gamma)})
    + 2\LCM_{(\alpha}(\beta_{\beta}n_{\gamma)})
    \right)
     \\\nonumber & =   n^\alpha n^\beta h^\gamma_\tau \LCM_{(\alpha} \K_{\beta\gamma)}+n^\alpha n^\beta h^\gamma_\tau \LCM_{\gamma} (\zeta)n_{(\alpha} n_{\beta)}+2\zeta n^\alpha n^\beta h^\gamma_\tau\LCM_{(\alpha} ( n_{\beta}) n_{\gamma)}\\\nonumber 
    &+ 2  n^\alpha n^\beta h^\gamma_\tau\LCM_{(\alpha}(\beta_{\beta})n_{\gamma)} + 2  n^\alpha n^\beta h^\gamma_\tau\beta_{(\alpha}\LCM_{\beta}(n_{\gamma)}) \\\nonumber  & =  n^\alpha n^\beta h^\gamma_\tau \LCM_{(\alpha} \K_{\beta\gamma)}+2 \LCS_\tau (\zeta)-4 \zeta \LCS_{\tau} \ln \varphi+4 \epsilon h^\gamma_\tau n^\alpha \LCM_{\alpha}(\beta_{\gamma})-4 \epsilon \SFS^\alpha_\tau\beta_{\alpha}\\\nonumber  & = 4 \epsilon  \K^\beta_{\tau}\LCS_\beta \ln \varphi+2 \LCS_\tau (\zeta)-4 \zeta \LCS_{\tau} \ln \varphi+4 \epsilon h^\gamma_\tau n^\alpha \LCM_{\alpha}(\beta_{\gamma})-4 \epsilon \SFS^\alpha_\tau\beta_{\alpha},
\end{align}
\begin{align}
 n^\alpha n^\beta n^\gamma
    \LCM_{(\alpha} K_{\beta\gamma)} &=6 
    n^\alpha n^\beta n^\gamma
    \LCM_{\alpha} \left(
      \K_{\beta\gamma}
    + \zeta n_{\beta} n_{\gamma}
    + 2\beta_{\beta}n_{\gamma}
    \right)
    = 6  n^\alpha \LCM_{\alpha} \zeta+ 12 \epsilon 
    n^\alpha n^\beta 
    \LCM_{\alpha}(\beta_{\beta}) \\\nonumber & = 
    6  n^\alpha \LCM_{\alpha} \zeta+12 
    \beta^{\beta} \LCS_\beta \ln \varphi,
\end{align}
\end{subequations} 
where we have used the following identities
\begin{subequations} 
\begin{align} 
 n^\alpha h^\beta_\sigma h^\gamma_\tau
      \LCM_{(\alpha} \K_{\beta\gamma)}&= h^\beta_{(\sigma} h^\gamma_{\tau)}
      n^\alpha \LCM_{\alpha} \K_{\beta\gamma}- 2\LCM_{\beta} (n^\alpha) h^\beta_\sigma 
      \K_{\alpha\tau}-2\LCM_{\gamma}(n^\alpha)  h^\gamma_\tau
      \K_{\alpha\sigma}\\ &= h^\beta_{(\sigma} h^\gamma_{\tau)}
      n^\alpha \LCM_{\alpha} \K_{\beta\gamma}- 2\SFS^\alpha_{(\sigma} 
      \K_{\tau)\alpha}\nonumber,\\
 n^\alpha n^\beta h^\gamma_\tau \LCM_{(\alpha} \K_{\beta\gamma)}&=4 n^\alpha n^\beta h^\gamma_\tau \LCM_{\alpha} \K_{\beta\gamma}=4 \epsilon  \K^\beta_{\tau}\LCS_\beta \ln \varphi.
 \end{align}
 \end{subequations} 
The projections of the right hand side of Eq. (\ref{eq:killing_equation_b}) can be obtained as follows
 \begin{subequations} 
 \begin{align} 
 h^\alpha_\rho h^\beta_\sigma h^\gamma_\tau
    \Omega_{(\alpha} g_{\beta\gamma)}= {}^\tau\Omega_{(\rho} h_{\sigma\tau)}, \quad  n^\alpha h^\beta_\sigma h^\gamma_\tau \Omega_{(\alpha} g_{\beta\gamma)}={}^n\Omega h_{(\sigma \tau )}, \\
     n^\alpha n^\beta h^\gamma_\tau \Omega_{(\alpha} g_{\beta\gamma)}=2\epsilon \cdot{}^\tau\Omega_\tau, \quad  n^\alpha n^\beta n^\gamma \Omega_{(\alpha} g_{\beta\gamma)}=6 \epsilon\cdot{}^n\Omega.
\end{align}
 \end{subequations}

\section{Proof of proposition \ref{prop:killing_mixed}}
\label{proof_proposition_3}

If projection $\K_{\alpha\beta}$ is a CKT in slices, Eq. (\ref{eq:killing_map_xy_a}) results in the following equation
\begin{equation} \label{eq:appendix_1}
 \beta_{(\alpha} \SFS_{\beta}{}_{\gamma)} = \alpha'_{(\alpha} h_{\beta}{}_{\gamma)}.
\end{equation}
where $\alpha'^\alpha$ is some vector. This equation has a trivial solution $\beta_{\alpha}=\alpha'_\alpha=0$ with an arbitrary second fundamental form $\SFS_{\beta}{}_{\gamma}$, which constitutes case (a) of the proposition. Consider the case $\beta_{\alpha}\neq0$. First of all, we divide $\SFS_{\beta}{}_{\gamma}$ into a trace and a trace-free parts
\begin{equation} 
\SFS_{\beta}{}_{\gamma}=\frac{\SFS_{\alpha}{}^{\alpha}}{m-1} h_{\beta\gamma} + \sigma_{\beta}{}_{\gamma}.
\end{equation}
If $\sigma_{\alpha\beta}=0$, the given slices are umbilic. Then, instead of Eq. (\ref{eq:appendix_1}) we get
\begin{equation} \label{eq:appendix_3}
 \beta_{(\alpha} \sigma_{\beta}{}_{\gamma)} = \alpha_{(\alpha} h_{\beta}{}_{\gamma)}, \quad \alpha_\alpha\equiv \alpha'_\alpha-\frac{1}{m-1}\SFS_{\alpha}{}^{\alpha}\beta_\alpha,
\end{equation}
Form the trace of Eq. (\ref{eq:appendix_3}) one gets the expression for $\alpha_\gamma$
\begin{equation} \label{eq:appendix_2} 
 \alpha_{\gamma} = \frac{2}{m+1} \beta^{\alpha} \sigma_{\alpha}{}_{\gamma}.
\end{equation}
On the other hand, contraction of Eq. (\ref{eq:appendix_3}) with $\beta^{\alpha}$ once, twice and thrice gives
\begin{subequations}\label{eq:appendix_4}
\begin{align}  
&(\beta^{\alpha} \beta_{\alpha}) \sigma_{\beta}{}_{\gamma}= (\beta^{\alpha}\alpha_{\alpha}) h_{\beta}{}_{\gamma} -\frac{m-1}{2}\alpha_{(\beta} \beta_{\gamma)},\label{eq:appendix_4_a}\\
  &(\beta^{\beta}\beta_{\beta})\alpha_{\gamma} = \frac{3-m}{2m}(\beta^{\alpha}\alpha_{\alpha}) \beta_{\gamma},\label{eq:appendix_4_b}\\
 & (m-1)(\beta^{\beta}\beta_{\beta})(\beta^{\gamma}\alpha_{\gamma} )=0,\label{eq:appendix_4_c}
\end{align}
\end{subequations}
where we have used Eq. (\ref{eq:appendix_2}).

\textbf{Case A.} If $\beta^{\alpha}\beta_{\alpha}\neq 0$, there is a condition $\beta^{\gamma}\alpha_{\gamma} =0$ from Eq. (\ref{eq:appendix_4_c}), implying a condition $\alpha_{\gamma} =0$ from Eq. (\ref{eq:appendix_4_b}), and $\sigma_{\beta}{}_{\gamma}=0$ from Eq. (\ref{eq:appendix_4_a}). This means that slices are totally umbilic, $\SFS_{\alpha\beta}=\SFS_{\gamma}{}^\gamma h_{\alpha\beta}/(m-1)$. 

\textbf{Case B.} If $\beta^{\alpha}\beta_{\alpha}=0$, Eq. (\ref{eq:appendix_4_b}) reduces to $(3-m)\beta^{\alpha}\alpha_{\alpha} = 0$, which can be satisfied in the following two subcases:

\textbf{Case BI.} If $\beta^{\alpha}\alpha_{\alpha}=0$, Eq. (\ref{eq:appendix_4_a}) implies $\alpha_{(\beta} \beta_{\gamma)}=0$ and contracting this with $\alpha^\beta$, vector $\alpha^\beta$ is null $\alpha_\alpha \alpha^\alpha=0$.
That is, $\beta^{\alpha}$ and $\alpha^{\alpha}$ are orthogonal to each other, null and linearly independent, i.e., $\alpha^{\alpha}=0$ since the manifold is Lorentzian. From Eq. (\ref{eq:appendix_2}) and Eq. (\ref{eq:appendix_3}) follows
\begin{equation} \label{eq:app_A_case_BII_sigma}
    \beta^{\beta}\sigma_{\beta}{}_{\alpha} = 0, \qquad
    \beta_{(\alpha} \sigma_{\beta}{}_{\gamma)} = 0.
\end{equation}
Let us introduce a complementary null vector ${\beta'}^\alpha$ such that $\beta'_{\alpha}\beta^{\alpha}=-1$, $\beta'^{\alpha}\beta'_{\alpha}=0$ and other orthogonal complementary vectors $X_{(i)}^\alpha$ such that $X^\alpha_{(i)}\beta_{\alpha}=X^\alpha_{(i)}\beta'_{\alpha}=0$, $X^\alpha_{(i)}X_{(j)\alpha} = \delta_{ij}$. Contraction of the second equation in (\ref{eq:app_A_case_BII_sigma}) with any pair or triplet of vectors from the set (and using the first one) leaves us with a conclusion that any component of $\sigma_{\alpha\beta}$ is zero. Thus, the slices are totally umbilic $\sigma_{\alpha}{}_{\beta}=0$. 

\textbf{Case BII.} The case $m=3$, $\beta^{\alpha}\alpha_{\alpha}\neq0$ solves Eq. (\ref{eq:appendix_4_b}), and from Eq. (\ref{eq:appendix_4_a}) follows
\begin{align} 
    h_{\beta}{}_{\gamma} = 
    \alpha_{(\beta} \beta_{\gamma)}/(\beta^{\alpha}\alpha_{\alpha})
    \qquad\Rightarrow\qquad
    \alpha_{\gamma} = 
    \beta_{\gamma} (\alpha^{\beta}\alpha_{\beta})/(\beta^{\alpha}\alpha_{\alpha})+\alpha_{\gamma},
\end{align}
which is possible only if $\alpha^{\beta}\alpha_{\beta}=0$. Contracting Eqs. (\ref{eq:appendix_3}), (\ref{eq:appendix_2}) with a certain combinations of $\alpha^{\alpha}$ and $\beta^{\alpha}$, all independent components of $\sigma_{\alpha\beta}$ can be extracted
\begin{align} 
\alpha^{\beta}\sigma_{\beta}{}_{\gamma}\alpha^{\gamma}=0,\qquad
\alpha^{\beta}\sigma_{\beta}{}_{\gamma}\beta^{\gamma}=0,\qquad
\beta^{\alpha} \sigma_{\alpha}{}_{\gamma} \beta^{\gamma}= 2 \alpha_{\gamma}\beta^{\gamma}.
\end{align}
In the holonomic basis restrictions on $\sigma_{\beta}{}_{\gamma}$, $h_{\beta}{}_{\gamma}$ reads
\begin{align}  \label{eq:appendix_kk}
&h_{\alpha\beta}=\alpha_{(\alpha} \beta_{\beta)}/F, \quad \sigma_{\alpha\beta}= 2\alpha_{\alpha} \alpha_{\beta}/F,
\end{align}
where $F$ is an arbitrary function defining the scalar product $\alpha_\gamma \beta^\gamma = F$. It is easy to check that this solution does satisfy (\ref{eq:appendix_3}) and the slice is not umbilic ($\sigma_{\alpha\beta}\neq0$). Note that in higher dimensions $m>3$ we can define the same formal solution, but the corresponding metric will be degenerate (such that the rank of the matrix $h_{\alpha\beta}$ remains equal to two). Since the slice has the signature $(1,1)$ and vectors $\alpha^\alpha$, $\beta^\alpha$ should be null and not unidirectional, the direction of $\alpha^\alpha$ is uniquely defined through $\beta^\alpha$ if $h_{\alpha\beta}$ is known. In addition, $h_{\alpha\beta}$ and $\chi_{\alpha\beta}$ must be geometrically consistent.

\section{Proof of proposition \ref{prop:second_form}}
\label{proof_proposition_5}

Contracting the second fundamental form $\SFS_{\alpha\beta}$ with conformal Killing vectors $\mathcal{K}_{a}{}^{\alpha}$, we find
\begin{align}
    &
   -2\SFS_{\alpha\beta} \mathcal{K}_{a}{}^{\alpha}\mathcal{K}_{b}{}^{\beta}=  n_\beta \mathcal{K}_a{}^{\alpha} \nabla_\alpha \mathcal{K}_b{}^{\beta}
    + n_\beta \mathcal{K}_b{}^{\alpha} \nabla_\alpha \mathcal{K}_a{}^{\beta}
    = \\\nonumber & =
      2\Omega_b n_\beta \mathcal{K}_a{}^{\alpha} h_{\alpha}^{\beta}
    + 2\Omega_a n_\beta \mathcal{K}_b{}^{\alpha} h_{\alpha}^{\beta}
    - n_\beta \mathcal{K}_a{}^{\alpha} \nabla^\beta \mathcal{K}_b{}_{\alpha}
    - n_\beta \mathcal{K}_b{}^{\alpha} \nabla^\beta \mathcal{K}_a{}_{\alpha}
    = \\\nonumber & = 
    - n_\beta \left(
      \mathcal{K}_a{}^{\alpha} \nabla^\beta \mathcal{K}_b{}_{\alpha}
    + \mathcal{K}_b{}^{\alpha} \nabla^\beta \mathcal{K}_a{}_{\alpha}
    \right)
    = - n_\beta \nabla^\beta \mathcal{K}_a{}^{\alpha}\mathcal{K}_b{}_{\alpha}
    = - n_\beta \nabla^\beta \G_{ab},
\end{align}
where in the first row we have used the property of the second fundamental form, in the second row we have used the definition of the CKTs, in the third row we have used the orthogonality between the Killing vectors and the normal vector, and the Leibniz rule.

\section{Proof of proposition \ref{prop:reducible_umbilic}}
\label{proof_proposition_4}

Substituting CKT $K^{\alpha\beta} = e^\Psi n^\alpha n^\beta$ with $\Omega_\alpha = 2 \chi e^{\Psi} n_\alpha$ back into Killing equations (\ref{eq:killing_equation_b}) and contracting only one index with $n^\gamma$, one can get
\begin{align}
     &  n^\gamma \LCM_{(\alpha}  (e^{\Psi / 2} n_{\beta} ) n_{\gamma)}= \chi e^{\Psi / 2} n^\gamma  n_{(\alpha}g_{\beta\gamma)} \quad \Rightarrow\\
       & \epsilon\LCM_{(\alpha}  (e^{\Psi / 2} n_{\beta)} )+ e^{\Psi / 2} n^\lambda \LCM_{\lambda} n_{(\alpha}  n_{\beta)} +\epsilon   \LCM_{(\alpha}  (e^{\Psi / 2})  n_{\beta)}+ e^{\Psi / 2} n^\gamma \LCM_{\gamma} (\Psi) n_{\alpha}  n_{\beta} \nonumber\\&= 4\chi e^{\Psi / 2}   n_{\alpha}n_{\beta}+2\epsilon\chi e^{\Psi / 2} g_{\alpha\beta} \quad \Rightarrow\nonumber\\
        &\epsilon\LCM_{(\alpha}  ( e^{\Psi / 2} n_{\beta)} )+e^{\Psi / 2}\underbrace{\left(\frac{\epsilon}{2}  \LCM_{(\alpha}  \Psi +  n^\lambda \LCM_{\lambda} n_{(\alpha}-\chi   n_{(\alpha}\right)}_{0}n_{\beta)} =  2\epsilon\chi e^{\Psi / 2} g_{\alpha\beta}.\nonumber
\end{align}
The expression in brackets is zero since the following identity follows from Eq. (\ref{eq:kk_eq_psi})
\begin{align}
\LCM_{\alpha}  \Psi=\LCS_{\alpha}  \Psi+ \epsilon n_\alpha n^\beta \LCM_{\beta}\Psi=-2\epsilon n^\lambda \LCM_{\lambda} n_{\alpha} +  2 \epsilon \chi n_\alpha.
\end{align}

\end{document}